\documentclass[draftcls, onecolumn, 11pt]{IEEEtran}

\usepackage[T1]{fontenc}
\usepackage{graphicx}
\usepackage{amsthm}
\usepackage{amsmath}
\usepackage{epstopdf}
\usepackage[all]{xy}
\usepackage[cmintegrals]{newtxmath}

\newtheorem{proposition}{Proposition}
\newtheorem{corollary}{Corollary}
\newtheorem{definition}{Definition}
\newtheorem{remark}{Remark}
\newtheorem{theorem}{Theorem}
\newtheorem{example}{Example}
\newtheorem{lemma}{Lemma}
\newtheorem{acknowledgements}{Acknowledgements}

\newcommand{\rad}{\mbox{rad}(\nu) \mid \mbox{rad}(n)}
\newcommand{\F}{\mathbb{F}}
\newcommand{\Fq}{\mathbb{F}_{q}}
\newcommand{\Fqt}{\tilde{\mathbb{F}}_{q}}
\newcommand{\Z}{\mathbb{Z}}
\newcommand{\tree}[3]{\mbox{Tree}_{#1}( #2 / #3 ) }
\newcommand{\pred}[3]{\mbox{Pred}_{#1}( #2 / #3 ) }
\newcommand{\G}[2]{\mathcal{G}(#1 / #2 )}
\newcommand{\GT}[1]{\mathcal{G}(T_n / #1 )}
\newcommand{\mo}[1]{\mbox{ord}(#1)}

\newcommand{\cyc}[1]{\mbox{Cyc}(#1,\bullet)}
\newcommand{\cycc}[2]{\mbox{Cyc}(#1, \langle #2 \times \bullet   \rangle)}

\newcommand{\s}[2]{\sum_{i=1}^{#2 -1} {#1}_{1}\cdots {#1}_{i}}

\ifCLASSINFOpdf
\else
\fi

\begin{document}

\title{The Graph Structure of Chebyshev Polynomials  
over Finite Fields and Applications}


\author{Claudio~Qureshi and Daniel Panario
\thanks{Claudio Qureshi is with the Institute of Mathematics, Statistics and Computing Science of the University of Campinas, SP , Brazil (email: {cqureshi{@}ime.unicamp.br}) and Daniel Panario is with School of Mathematics and Statistics, 
Carleton University, Canada (email: {daniel{@}math.carleton.ca)}}}

%

%


\maketitle

\begin{abstract}
We completely describe the functional graph associated to 
iterations of Chebyshev polynomials over finite fields. 
Then, we use our structural results to obtain estimates 
for the average rho length, average number of connected 
components and the expected value for the period and 
preperiod of iterating Chebyshev polynomials.
\end{abstract}

\section{Introduction}\label{introduction}
The iteration of polynomials and rational functions over finite 
fields have recently become an active research topic. These 
dynamical systems have found applications in diverse areas,
including cryptography, biology and physics. In cryptography,
iterations of functions over finite fields were popularized by 
the Pollard rho algorithm for integer factorization \cite{P75}; 
its variant for computing discrete logarithms is considered 
the most efficient method against elliptic curve cryptography 
based on the discrete logarithm problem \cite{P78}. Other 
cryptographical applications of iterations of functions include 
pseudorandom bit generators \cite{BBS86}, and integer 
factorization and primality tests \cite{L30,L78}. 

When we iterate functions over finite structures, there is an 
underlying natural functional graph. For a function $f$ over a 
finite field $\F_q$, this graph has $q$ nodes and a directed 
edge from vertex $a$ to vertex $b$ if and only if $f(a)=b$. 
It is well known, combinatorially, that functional graphs are 
sets of connected components, components are directed cycles of 
nodes, and each of these nodes is the root of a directed tree 
from leaves to its root; see, for example, \cite{FlajOdl1990}. 

Some functions over finite fields when iterated present strong
symmetry properties. These symmetries allow mathematical proofs
for some dynamical properties such as period and preperiod of 
a generic element, (average) ``rho length'' (number of iterations 
until cycling back), number of connected components, cycle 
lengths, etc. In this paper we are interested on these kinds of 
properties for Chebyshev polynomials over finite fields, closely 
related to Dickson polynomials over finite fields. These 
polynomials, specially when they permute the elements of the 
field, have found applications in many areas including 
cryptography and coding theory. See \cite{LMT1993} for a 
monograph on Dickson polynomials and their applications,
including cryptography; for a more recent account on research 
in finite fields including Dickson polynomials, see \cite{HFF}. 

Previous results for quadratic functions are in \cite{VS04}; 
iterations of $x+x^{-1}$ have been dealt in \cite{Ugolini} and 
iterations of R{\'e}dei functions over non-binary finite fields 
appeared in \cite{QP15,QPM17}. Related to this paper, 
iterations of Chebyshev polynomials over finite fields have 
been treated in \cite{Gassert14}. The graph and 
periodicity properties for Chebyshev polynomials over finite 
fields when the degree of the polynomial is a prime number are 
given in \cite{Gassert14}. 

In this paper we study the action of Chebyshev functions of 
\emph{any} degree over finite fields. We give a structural 
theorem for the functional graph from which it is not hard to 
derive many periodicity properties of these iterations. In 
the literature there are two kinds of Chebyshev polynomials:
normalized and not normalized. We use the latter ones,
generally known as Dickson polynomials of the first kind. 
In odd characteristic both kinds of Chebyshev polynomials 
are conjugates of each other, and so their functional graphs 
are isomorphic. However, this is not the case in even 
characteristic. Using the normalized version trivializes 
since we get $T_n(x) =1$ if $n$ is even, and $T_n(x)=x$ if 
$n$ is odd, where $T_n$ is the $n$th degree Chebyshev polynomial.
As a consequence, we work with the non normalized version 
that is much richer in characteristic $2$. Not much is known 
about Chebyshev polynomials over binary fields; see 
\cite{FL16} for results over the $2$-adic integers.

In Section~\ref{preliminaries} we introduce relevant concepts
for this paper like $\nu$-series and their associated trees.
These trees play a central role in the description of the 
Chebyshev functional graph. Several results about a 
homomorphism of the Chebyshev functional graph, as well as 
a relevant covering notion, are given in Section 
\ref{homofunctionalgraphs}. A decomposition of the 
Chebyshev's functional graph is given in 
Section~\ref{uniformcomponents}. This decomposition leads 
naturally into three parts: the rational, the quadratic 
and the special component. Section \ref{rationalquadratic}
treats the rational and quadratic components. The special 
component is dealt in Section \ref{special}. The main result 
of this paper (Theorem \ref{ThMain}), a structural theorem 
for Chebyshev polynomials, is given in Section~\ref{structural}. 
We provide several examples to show applications of our main 
theorem. As a consequence of our main structural theorem, in 
this section we also obtain exact results for the parameters 
$N,C,T_0,T$ and $R$ for Chebyshev polynomials, where $N$ is 
the number of cycles (that is, the number of connected 
components), $T_0$ is the number of cyclic (periodic) points, 
$C$ is the expected value of the period, $T$ is the expected 
value of the preperiod, and $R$ is the expected rho length.

\section{Preliminaries} \label{preliminaries}

We denote by $\F_q$ a finite field with $q$ element, where $q$ 
is a prime power, and $\Z_d$ the ring of integers modulo $d$. 
Let $\F_q^{*}$ and $\Z_d^{*}$ denote the multiplicative 
group of inverse elements of $\F_q$ and $\Z_d$, respectively. 
Let $\overline{n}$ denote the equivalence class of $n$ modulo $d$. 
For $n,d \in \Z^{+}$ with $\gcd(n,d)=1$, we denote by $o_{d}(n)$ 
and $\tilde{o}_{d}(n)$ the multiplicative order of $\overline{n}$ 
in $\Z_d^{*}$ and $\Z_d^{*}/\{1,-1\}$, respectively. It is easy 
to see that if $-\overline{1} \in \langle \overline{n} \rangle$ 
in $\Z_d^{*}$, then $\tilde{o}_{d}(n)=o_d(n)/2$, otherwise 
$\tilde{o}_{d}(n)=o_{d}(n)$. For $m \in \Z^{+}$ we denote by 
$\mbox{rad}(m)$ the radical of $m$ which is defined as the product 
of the distinct primes divisors of $m$. We can decompose $m=\nu\omega$ 
where $\mbox{rad}(\nu)\mid \mbox{rad}(n)$ and $\gcd(\omega,n)=1$ which 
we refer as the \emph{$n$-decomposition} of $m$. If $f: X \rightarrow X$ 
is a function defined over a finite set $X$, we denote by 
$\mathcal{G}(f/X)$ its functional graph.

The main object of study of this paper is the action of Chebyshev 
polynomials over finite fields $\F_q$. The Chebyshev polynomial 
of the first kind of degree $n$ is denoted by $T_n$. This is the 
only monic, degree-$n$ polynomial with integer coefficients 
verifying $T_n(x+x^{-1})=x^{n}+x^{-n}$ for all $x\in \Z$. Table \ref{examples} gives the first Chebyshev polynomials.

\begin{table}[ht]
\begin{tabular}{ll}
$T_{ 1 }(x)$ & $= x $\\
$T_{ 2 }(x)$ & $= x^{2} - 2 $\\
$T_{ 3 }(x)$ & $= x^{3} - 3 x $\\
$T_{ 4 }(x)$ & $= x^{4} - 4 x^{2} + 2 $\\
$T_{ 5 }(x)$ & $= x^{5} - 5 x^{3} + 5 x $\\
$T_{ 6 }(x)$ & $= x^{6} - 6 x^{4} + 9 x^{2} - 2 $\\
$T_{ 7 }(x)$ & $= x^{7} - 7 x^{5} + 14 x^{3} - 7 x $\\
$T_{ 8 }(x)$ & $= x^{8} - 8 x^{6} + 20 x^{4} - 16 x^{2} + 2 $\\
$T_{ 9 }(x)$ & $= x^{9} - 9 x^{7} + 27 x^{5} - 30 x^{3} + 9 x $\\
$T_{ 10 }(x)$ & $= x^{10} - 10 x^{8} + 35 x^{6} - 50 x^{4} + 25 x^{2} - 2 $\\
\end{tabular}
\vspace{3mm}
\caption{First few Chebyshev polynomials $T_{n}(x)$ for $1\leq n \leq 10$.}\label{examples}
\end{table}

A remarkable property of these polynomials is that
$T_{n}\circ T_{m}=T_{nm}$ for all $m,n \in \Z^{+}$. In particular, 
$T_{n}^{(k)}=T_{n^{k}}$, where $f^{(k)}$ denotes the composition 
of $f$ with itself $k$ times. Describing the dynamics of the 
Chebyshev polynomial $T_n$ acting on the finite field $\F_q$ is 
equivalent to describing the Chebyshev's graph $\mathcal{G}(T_n/\F_q)$.

The case when $n=\ell$ is a prime number was dealt by Gassert; 
see \cite[Theorem 2.3]{Gassert14}. In this paper we extend these 
results for any positive integer $n$. 

\begin{example}
For $n=30$ the corresponding Chebyshev polynomial is given by 
$T_{30}(x) = x^{30} - 30 x^{28} + 405 x^{26} - 3250 x^{24} 
 + 17250 x^{22} - 63756 x^{20} + 168245 x^{18} - 319770 x^{16} 
 + 436050 x^{14} - 419900 x^{12} + 277134 x^{10} - 119340 x^{8} 
 + 30940 x^{6} - 4200 x^{4} + 225 x^{2} - 2$. The graphs 
$\mathcal{G}(T_{30}/\F_{q})$ for $q=19$ and $q=23$ are shown  
in Fig.~\ref{Figcheby30q19q23}.
\end{example}

\begin{figure}[h]
\centering
\includegraphics[width=1.0\textwidth]{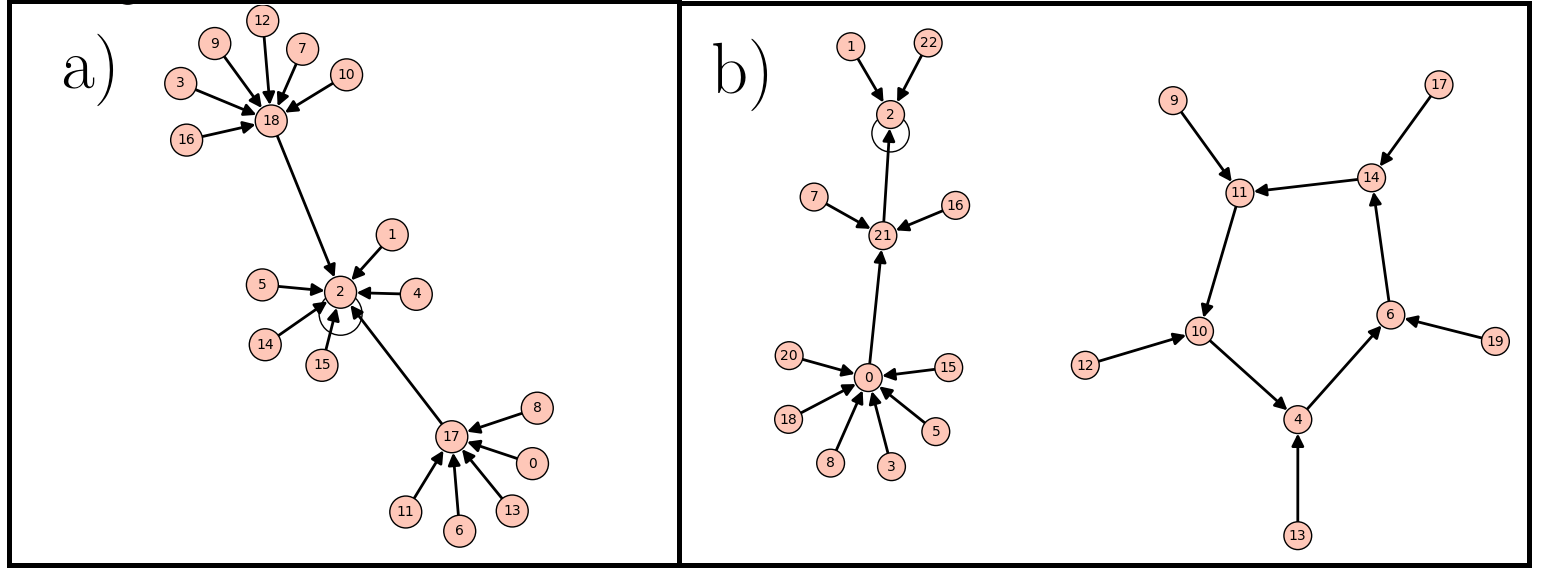}
\caption{a) The Chebyshev's graph $\mathcal{G}(T_{30}/\F_{19})$. b) The Chebyshev's graph $\mathcal{G}(T_{30}/\F_{23})$.} \label{Figcheby30q19q23}
\end{figure}

Next we review some concepts from \cite{QP15}. For $n$ and 
$\nu$ positive integers such that $\mbox{rad}(\nu)\mid \mbox{rad}(n)$, 
the $\nu$-series associated with $n$ is the finite sequence 
$\nu(n):=(\nu_1,\ldots,\nu_D)$ defined by the recurrence $\nu_1=\gcd(\nu,n), \nu_{k+1}=\gcd\left(\frac{\nu}{\nu_1\nu_2\cdots\nu_{k}}, n\right)$ 
for $1\leq k <D$ and $\nu_1\nu_2\cdots\nu_D =\nu$ with $\nu_D>1$ if $\nu>1$, and $\nu(n)=(1)$ if $\nu=1$.

We write $A=\biguplus B_i$ to indicate that $A$ is the union of 
pairwise disjoint sets $B_i$. If $m\in \Z^+$ and $T$ is a 
rooted tree, $\mbox{Cyc}(m,T)$ denotes a graph with a unique 
directed cycle of length $m$, where every node in this cycle 
is the root of a tree isomorphic to $T$. We also consider 
the disjoint union of the graphs $G_1,\ldots,G_k$, denoted 
by $\bigoplus_{i=1}^{k}G_{i}$, and $k\times G = 
\bigoplus_{i=1}^{k}G$ for $k\in \Z^{+}$. If $T_1,\ldots, T_k$ 
are rooted trees, $\langle T_1\oplus \cdots \oplus T_k \rangle$ 
is a rooted tree such that its root has exactly $k$ predecessors 
$v_1,\ldots,v_k$, and $v_i$ is the root of a tree isomorphic 
to $T_i$ for $i=1,\ldots,k$. If $T$ is a tree that consists of 
a single node we simply write $T = \bullet$. In particular, 
$\mbox{Cyc}(m,\bullet)$ denotes a directed cycle with $m$ nodes. The empty graph, denoted by $\emptyset$, is characterized by the properties: $\emptyset \oplus G=G$ for all graphs $G$, $k\times \emptyset = \emptyset$ for all $k\in \Z^{+}$ and $\langle
\emptyset \rangle = \bullet$.

We associate to each $\nu$-series $\nu(n)$ a rooted tree, 
denoted by $T_{\nu(n)}$, defined by the recurrence formula 
(see Fig.~\ref{FigTreeVseries}):
\begin{equation}\label{TreeAssociatedEq}
\left\{ \begin{array}{l}
           T^{0}= \bullet,   \\
          T^{k}= \langle \nu_k \times T^{k-1} \oplus 
                 \bigoplus_{i=1}^{k-1}(\nu_{i}-\nu_{i+1})\times T^{i-1} 
                 \rangle, 1\leq i <D, \\
    T_{\nu(n)} = \langle (\nu_D-1) \times T^{D-1} \oplus 
                 \bigoplus_{i=1}^{D-1}(\nu_{i}-\nu_{i+1})\times T^{i-1} 
                 \rangle.
         \end{array}
\right.\end{equation}
The tree $T_{\nu(n)}$ has $\nu$ vertices and depth $D$; see Proposition 2.14 and Theorem 3.16 of \cite{QP15}. 

The following theorem is a direct consequence of Corollary 3.8 and Theorem 3.16 of \cite{QP15}. As usual, $\varphi$ denotes Euler's totient function.

\begin{theorem} \label{ThNmapIsom}
Let $n\in \Z^{+}$ and $m=\nu\omega$ be the $n$-decomposition of $m$. Denoting by $\mathcal{G}(n/\Z_m)$ the functional graph of the multiplication-by-$n$ map on the cyclic group $\Z_m$, the following isomorphism holds:
$$ \mathcal{G}(n/\Z_m) = \bigoplus_{d\mid \omega} \frac{\varphi(d)}{{o}_{d}(n)}\times\mbox{Cyc}\left({o}_{d}(n),T_{\nu(n)}\right).   $$
\end{theorem}

\begin{figure}[h]
\centering
\includegraphics[width=1.0\textwidth]{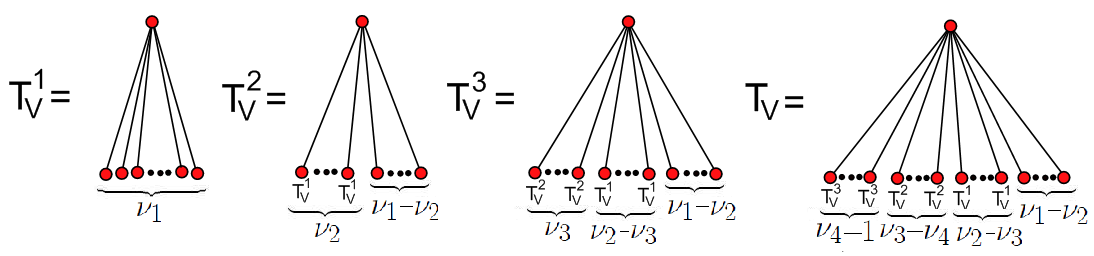}
\caption{This figure (taken from \cite{QP15}) illustrates the 
inductive definition of $T_V$ when $V$ is a $\nu$-series with 
four components $V=(\nu_1,\nu_2,\nu_3,\nu_4)$. A node $v$ 
labelled by a rooted tree $T$ indicates that $v$ is the root 
of a tree isomorphic to $T$.} \label{FigTreeVseries}
\end{figure}

A strategy to 
describe a functional graph $\mathcal{G}(f/X)$ of a function 
$f:X\rightarrow X$ is decomposing the set $X$ in $f$-invariant 
components. A subset $A\subseteq X$ is \emph{forward $f$-invariant} 
when $f(A)\subseteq A$. In this case the graph $\mathcal{G}(f/A)$ 
is a subgraph of $\mathcal{G}(f/X)$. If $f^{-1}(A)\subseteq A$, 
the set $A$ is \emph{backward $f$-invariant}. The set $A$ is 
\emph{$f$-invariant} if it is both forward and backward $f$-invariant. 
In this case $\mathcal{G}(f/A)$ is not only a subgraph of 
$\mathcal{G}(f/X)$ but also a union of connected components 
and we can write $\mathcal{G}(f/X)=\mathcal{G}(f/A) \oplus 
\mathcal{G}(f/A^{c})$, where $A^{c}=X \setminus A$. In this 
paper, we decompose the set $\F_{q}$ in $T_n$-invariant subsets 
$A_1,\ldots, A_\kappa$ such that each functional graph 
$\mathcal{G}(T_n/A_i)$ for $i=1,\ldots,\kappa$ is easier to 
describe than the general case and 
$\GT{\Fq} =\bigoplus_{i=1}^{\kappa}\mathcal{G}(T_n/A_i)$.

To describe a functional graph we need to describe not only the cyclic part but also the rooted trees attached to the periodic points. We introduce next some notation related to rooted trees (where the root is not necessarily a periodic point). Let $f:X\rightarrow X$, $x\in X$ and $N_f$ be the set of its non-periodic points. We define the set of predecessors of $x$ by 
$$\pred{x}{f}{X}=\{y\in N_{f}: 
  f^{(k)}(y)=x \textrm{ for some }k\geq 1\} \cup \{x\}.$$ 

We denote by $\tree{x}{f}{X}$ the rooted tree with root $x$, 
vertex set $V=\pred{x}{f}{X}$ and directed edges $(y,f(y))$ 
for $y \in V\setminus \{x\}$.

\section{Results on homomorphism of functional graphs}
\label{homofunctionalgraphs}

A directed graph is a pair $G=(V,E)$ where $V$ is the vertex set 
and $E\subseteq V\times V$ is the edge set. A homomorphism $\phi$ 
between two directed graphs $G_1=(V_1,E_1)$ and $G_2=(V_2,E_2)$, 
denoted by $\phi: G_1 \rightarrow G_2$, is a function 
$\phi: V_1 \rightarrow V_2$ such that if $(v,v') \in E_1$ then 
$(\phi(v),\phi(v')) \in E_2$. In the particular case of functional 
graphs, a homomorphism $\phi: \G{f_1}{X_1} \rightarrow \G{f_2}{X_2}$ 
is a function $\phi: X_1 \rightarrow X_2$ satisfying 
$\phi \circ f_1 = f_2 \circ \phi$, or equivalently such that 
the following diagram commutes
$$ \xymatrix{X_1 \ar[r]^{\phi}  & X_2  \\
  X_1 \ar[r]^{\phi}\ar[u]^{f_1} & X_2  \ar[u]_{f_2}}.$$
It is easy to prove by induction that the relation 
$\phi \circ f_1 = f_2 \circ \phi$ implies 
$\phi \circ f_1^{(k)} = f_2^{(k)} \circ \phi$ for all $k\geq 1$, 
that is, $\phi: \G{f_1^{(k)}}{X_1} \rightarrow \G{f_2^{(k)}}{X_2}$ 
is also a homomorphism for all $k\geq 1$. If in addition $\phi$ 
is bijective (as function from $X_1$ to $X_2$) then 
$\phi: \G{f_1}{X_1} \rightarrow \G{f_2}{X_2}$ is an isomorphism 
of functional graphs. In this case the functional graphs are the 
same, up to the labelling of the vertices. The main result of 
this paper (Theorem \ref{ThMain}) is an explicit description 
of $\GT{\F_q}$, the functional graph of the Chebyshev polynomial 
$T_n$ over a finite field $\Fq$. 

In the first part of this section we introduce the concept of 
$\theta$-covering between two functional graphs and derive some 
properties. In the last part we apply these results to obtain 
some rooted tree isomorphism formulas which are used in the 
next sections.

\subsection{$\theta$-coverings}
     
In our case of study (functional graph of Chebyshev polynomials) 
we consider the set $\Fqt = \Fq^{*}\cup H$, where $H$ is the 
multiplicative subgroup of $\F_{q^2}^{*}$ of order $q+1$, and 
the following maps: 
\begin{itemize}
\item The inversion map $i: \Fqt \rightarrow \Fqt$ given by 
$i(\alpha)=\alpha^{-1}$.
\item The exponentiation map $r_n : \Fqt \rightarrow \Fqt$ given 
by $r_n(\alpha)=\alpha^{n}$. 
\item The map $\eta: \Fqt \rightarrow \
\Fq$ given by $\eta(\alpha)=\alpha + \alpha^{-1}$.
\end{itemize}

A useful relationship between these maps and the Chebyshev map are 
$T_n \circ \eta = \eta \circ r_n$ and $r_n \circ i = i \circ r_n$. 
In other words we have the following commutative diagrams:

\begin{displaymath}\label{CommutativeDiagram}
\begin{array}{ccc}
\xymatrix{\Fqt \ar[r]^{\eta}  & \Fq  \\
     \Fqt \ar[r]^{\eta}\ar[u]^{r_n} & \Fq  \ar[u]_{T_n}} &  \textrm{\hspace{1cm}and\hspace{1cm}}    & \xymatrix{\Fqt \ar[r]^{i}  & \Fqt  \\
     \Fqt \ar[r]^{i}\ar[u]^{r_n} & \Fqt  \ar[u]_{r_n}}
\end{array}
\end{displaymath}

To describe the Chebyshev functional graph $\GT{\Fq}$ it is helpful to consider the homomorphism $\eta: \G{r_n}{\Fqt} \rightarrow \GT{\Fq}$ and to relate properties between these functional graphs. This homomorphism is not an isomorphism, but it has very nice properties that are captured in the next concept.

\begin{definition}
Let $\phi: \G{f_1}{X_1} \rightarrow \G{f_2}{X_2}$ be a homomorphism 
of functional graphs and $\theta: X_1 \rightarrow X_1$ be a 
permutation (bijection) which commutes with $f_1$ (that is, 
$f_1\circ \theta = \theta \circ f_1$). Then $\phi$ is a 
\emph{$\theta$-covering} if for every $a \in X_2$ there is 
$\alpha\in X_1$ such that 
$\phi^{-1}(a)= \{\theta^{(i)}(\alpha): i \in \Z \}$ (in other 
words, if the preimage of each point is a $\theta$-orbit). The homomorphism $\phi$ is a covering if it is a $\theta$-covering 
for some $\theta$ verifying the above properties. 
\end{definition}

We remark that a covering is necessarily onto and every isomorphism 
$\phi: \G{f_1}{X_1} \rightarrow \G{f_2}{X_2}$ is a covering (with 
respect to the identity map $\mbox{id}:X_1\rightarrow X_1$, 
$\mbox{id}(x)=x$). We note that the condition of $\phi^{-1}(a)$ 
being a $\theta$-orbit for all $a\in X_2$ implies that 
$\phi \circ \theta = \phi$. 

In \cite{Gassert14} it is proved several properties of the map 
$\eta$. Namely $\eta$ is surjective, $\eta^{-1}(2)=\{1\}$, 
$\eta^{-1}(-2)=\{-1\}$, and for $a\in \Fq$, 
$\eta^{-1}(a)=\{\alpha,\alpha^{-1}\}$ where $\alpha$ and 
$\alpha^{-1}$ are the roots (in $\F_{q^2}^{*}$) of $x^2-ax+1=0$ 
which are distinct if $a\neq \pm 2$. In particular, with our 
notation, we have that $\eta: \G{r_n}{\Fqt} \rightarrow \GT{\Fq}$ 
is a $i$-covering between these functional graphs.

Next we prove some general properties for coverings of functional 
graphs that are used in the next section for the particular case 
of the covering $\eta: \G{r_n}{\Fqt} \rightarrow \GT{\Fq}$. In the 
next propositions we denote by $P_f$ and $N_f$ the set of periodic 
and non-periodic points with respect to the map $f$, respectively. 
We note that if $\phi:\G{f_1}{X_1}\rightarrow\G{f_2}{X_2}$ is a 
homomorphism and $x\in P_{f_1}$ then there is a $k\geq 1$ such 
that $f_1^{(k)}(x)=x$. This implies 
$f_2^{(k)}(\phi(x))=\phi(f_1^{(k)}(x))=\phi(x)$, thus 
$x \in \phi^{-1}(P_{f_2})$ and we have 
$P_{f_1}\subseteq \phi^{-1}(P_{f_2})$. The next proposition 
shows that when $\phi$ is a covering this inclusion is in 
fact an equality.

\begin{proposition}\label{PropCoveringPreservePer}
Let $\theta:X_1\rightarrow X_1$ be a permutation satisfying 
$f_1\circ \theta = \theta\circ f_1$. If 
$\phi: \G{f_1}{X_1}\rightarrow \G{f_2}{X_2}$ is a 
$\theta$-covering then $\phi^{-1}(P_{f_2})=P_{f_1}$.
\end{proposition}

\begin{proof}
Let $\ell$ be the order of $\theta$ (i.e. $\theta^{(\ell)}=\mbox{id}$). 
It suffices to prove $\phi^{-1}(P_{f_2})\subseteq P_{f_1}$. If 
$\alpha \in \phi^{-1}(P_{f_2})$ then there is a $k\geq 1$ such that 
$f_2^{(k)}(\phi(\alpha))=\phi(\alpha)$. Since 
$f_2^{(k)}(\phi(\alpha))=\phi(f_1^{(k)}(\alpha))$ we conclude that 
$f_1^{(k)}(\alpha) = \theta^{(i)}(\alpha)$ for some $i \in \Z$. 
Applying $f_1^{(k)}$ on both sides we obtain 
$f_1^{(2k)}(\alpha)=f_1^{(k)}(\theta^{(i)}(\alpha))= 
\theta^{(i)}(f_1^{(k)}(\alpha)) = \theta^{(2i)}(\alpha)$. 
In the same way, applying $f_1^{(k)}$ several times, we have by 
induction that $f_1^{(mk)}(\alpha) = \theta^{(mi)}(\alpha)$ for 
all $m\geq 1$. With $m=\ell$ we obtain $f_1^{(\ell k)}(\alpha) 
= \theta^{(\ell i)}(\alpha)=\alpha$, thus $\alpha \in P_{f_1}$. 
\hfill $\qed$
\end{proof}

\begin{remark}\label{Remark1}
The equation $\phi^{-1}(P_{f_2})=P_{f_1}$ is equivalent to
$\phi^{-1}(N_{f_2})=N_{f_1}$ since $\phi^{-1}(X^{c})=\phi^{-1}(X)^{c}$.
\end{remark}

\begin{proposition}\label{PropPreservePerPreservePred}
Let $\phi : \G{f_1}{X_1}\rightarrow \G{f_2}{X_2}$ be a homomorphism 
satisfying $\phi^{-1}(P_{f_2})=P_{f_1}$ and $\alpha \in X_1$. We have 
$\pred{\alpha}{f_1}{X_1}\subseteq \phi^{-1}(\pred{\phi(\alpha)}{f_2}{X_2})$.
\end{proposition}

\begin{proof}
Let $\beta \in \pred{\alpha}{f_1}{X_1}$, $\beta \neq \alpha$ (in 
particular $\beta \in N_{f_1}$). By definition, there is an integer 
$k\geq 1$ such that $f_1^{(k)}(\beta)=\alpha$. This implies 
$f_2^{(k)}(\theta(\beta)) = \theta(f_1^{(k)}(\beta))=\theta(\alpha)$. 
Since $\phi^{-1}(N_{f_2})=N_{f_1}$ and $\beta \in N_{f_1}$ we have 
$\phi(\beta) \in N_{f_2}$, thus 
$\phi(\beta)\in \pred{\phi(\alpha)}{f_2}{X_2}$. \hfill $\qed$
\end{proof}

\begin{remark}\label{Remark2}
If $\pred{\alpha}{f_1}{X_1}\subseteq 
\phi^{-1}(\pred{\phi(\alpha)}{f_2}{X_2})$ then 
$\phi(\pred{\alpha}{f_1}{X_1})\subseteq 
\pred{\phi(\alpha)}{f_2}{X_2}$ since $\phi$ is surjective.
\end{remark}

\begin{proposition}\label{PropCoveringPreservePred2}
Let $\theta:X_1\rightarrow X_1$ be a permutation satisfying 
$f_1\circ \theta = \theta\circ f_1$, $\alpha \in X_1$ and 
$\phi: \G{f_1}{X_1}\rightarrow \G{f_2}{X_2}$ be a $\theta$-covering. 
The equality 
$\phi(\pred{\alpha}{f_1}{X_1})= \pred{\phi(\alpha)}{f_2}{X_2}$ holds.
\end{proposition}

\begin{proof}
The inclusion 
$\phi(\pred{\alpha}{f_1}{X_1})\subseteq \pred{\phi(\alpha)}{f_2}{X_2}$ 
follows from Propositions \ref{PropCoveringPreservePer} and 
\ref{PropPreservePerPreservePred} (see also Remark \ref{Remark2}). 
To prove the other inclusion we consider 
$b \in \pred{\phi(\alpha)}{f_2}{X_2}$ with $b\neq \phi(\alpha)$ 
(in particular $b\in N_{f_2}$) and $\beta \in X_1$ such that 
$b=\phi(\beta)$ (this is possible because $\phi$ is surjective). 
We have to prove that there is a point 
$\beta' \in \pred{\alpha}{f_1}{X_1}$ such that $\phi(\beta')=b$. 
By definition there is an integer $k\geq 1$ such that 
$f_2^{(k)}(b) = \phi(\alpha)$ and we have 
$\phi(f_1^{(k)}(\beta))= f_2^{(k)}(\phi(\beta)) = \phi(\alpha)$. 
Since $\phi$ is a $\theta$-covering, from $\phi(f_1^{(k)}(\beta))
= \phi(\alpha)$ we have that $\alpha = \theta^{(i)}(f_1^{(k)}(\beta))$ 
for some integer $i$ and define $\beta'=\theta^{(i)}(\beta)$. 
Using that $\theta$ and $f_1$ commute we obtain $f_1^{(k)}(\beta')
=\theta^{(i)}(f_1^{(k)}(\beta))=\alpha$ and $\phi(\beta')
=\phi(\theta^{(i)}(\beta))=\phi(\beta)=b$ (because 
$\phi\circ \theta = \phi$). To conclude the proof we have to show 
that $\beta' \in \pred{\alpha}{f_1}{X_1}$ and it suffices to prove 
that $\beta' \in N_{f_1}$. Since $\phi(\beta')=b\in N_{f_2}$ we 
have $\beta'\in \phi^{-1}(N_{f_2})=N_{f_1}$ by Proposition 
\ref{PropCoveringPreservePer} (see also Remark \ref{Remark1}). 
\hfill $\qed$
\end{proof}

With the same notation and hypothesis of Proposition
 \ref{PropCoveringPreservePred2}, if we denote by 
$P_1=\pred{\alpha}{f_1}{X_1}$ and $P_2=\pred{\phi(\alpha)}{f_2}{X_2}$ 
we have that the restricted function $\phi|_{P_1}:P_1 \rightarrow P_2$ 
is onto. We want to find conditions to guarantee that 
$\phi|_{P_1}:P_1 \rightarrow P_2$ is a bijection. We recall  
that the order of a permutation $\theta: X_1 \rightarrow X_1$ 
is the smallest positive integer $\ell \geq 1$ such that 
$\theta^{(\ell)}=\mbox{id}$. This implies that the cardinality 
of the $\theta$-orbit of a point $\alpha \in X_1$, given by 
$\{\theta^{(i)}(\alpha): 0 \leq i < \ell\}$, is a divisor of $\ell$.

\begin{definition}
Let $\theta:X_1 \rightarrow X_1$ be a permutation of order $\ell$. 
A point $\alpha \in X_1$ is \emph{$\theta$-maximal}, if the 
sequence of iterates:  $\alpha, \theta(\alpha), \theta^{(2)}(\alpha),
\ldots, \theta^{(\ell-1)}(\alpha)$ are pairwise distinct (that is, 
if the $\theta$-orbit of $\alpha$ has exactly $\ell$ elements).
\end{definition}

\begin{remark}
An important particular case is when $\theta: X_1 \rightarrow X_1$ 
is the identity map. In this case every point $\alpha \in X_1$ is 
$\theta$-maximal.
\end{remark}

\begin{proposition}\label{PropCoveringAndTreeIsom}
Let $\theta:X_1\rightarrow X_1$ be a permutation satisfying 
$f_1\circ \theta = \theta\circ f_1$, $\alpha$ be a $\theta$-maximal 
point of $X_1$ and $\phi: \G{f_1}{X_1}\rightarrow \G{f_2}{X_2}$ 
be a $\theta$-covering. We denote by $P_1=\pred{\alpha}{f_1}{X_1}$ 
and $P_2= \pred{\phi(\alpha)}{f_2}{X_2}$. Then the restricted map 
$\phi|_{P_1}:P_1 \rightarrow P_2$ is a bijection.
\end{proposition}

\begin{proof}
By Proposition \ref{PropCoveringPreservePred2} we have that 
$\phi|_{P_1}:P_1 \rightarrow P_2$ is onto. To prove that 
$\phi|_{P_1}$ is $1$-to-$1$ we consider $\beta_1,\beta_2 \in P_1$ 
such that $\phi(\beta_1)=\phi(\beta_2)$. Then there is an 
integer $i\in\Z$ such that $\beta_2 = \theta^{(i)}(\beta_1)$. If 
the order of the permutation $\theta$ is $\ell$, we can suppose that 
$0\leq i <\ell$ and we also have $\beta_1=\theta^{(\ell-i)}(\beta_2)$. 
We consider the smallest integers $s_1,s_2\geq 0$ such that 
$f_1^{(s_i)}(\beta_i)=\alpha$ for $i=1,2$ (they exist because 
$\beta_1,\beta_2 \in P_1$). We want to prove that $s_1=s_2$. 
Consider the smallest integer $t\geq 0$ such that 
$f_1^{(t)}(\alpha) \in P_{f_1}$. We have that 
$\theta: \G{f_1}{X_1} \rightarrow \G{f_1}{X_1}$ is an isomorphism 
of a functional graph (since $\theta$ is bijective and 
$\theta \circ f_1 = f_1 \circ \theta$), thus, by Proposition 
\ref{PropCoveringPreservePer}, $\theta^{-1}(P_{f_1})=P_{f_1}$. 
We have that 
$f_1^{(t+s_2)}(\beta_1) = \theta^{(\ell-i)}(f_1^{(t+s_2)}(\beta_2))
 = \theta^{(\ell-i)}(f_1^{(t)}(\alpha)) \in \theta^{(\ell-i)} (P_{f_1})
 =P_{f_1}$ (in particular $t+s_2 \geq s_1$ because 
$f_1^{(t+s_2)}(\beta_1) \in P_{f_1}$ and $\beta_1$ is a predecessor 
of $\alpha$). We have that 
$f_1^{(t+s_2-s_1)}(\alpha)=f_1^{(t+s_2-s_1)}(f_1^{s_1}(\beta_1))
 =f_1^{(t+s_2)}(\beta_1)\in P_{f_1}$ and by the minimality of $t$ 
we conclude that $s_2\geq s_1$. In a similar way we prove the other 
inequality $s_2\leq s_1$ obtaining $s_2=s_1$; let us  denote by 
$s=s_1=s_2$. We have 
$\alpha = f_1^{(s)}(\beta_2) = f_i^{(s)}(\theta^{(i)}(\beta_1))
 = \theta^{(i)}(f_1^{(s)}(\beta_1))=\theta^{(i)}(\alpha)$ with 
$0\leq i < \ell$. Using that $\alpha$ is $\theta$-maximal we 
conclude that $i=0$ and $\beta_1=\beta_2$ as desired. \hfill $\qed$
\end{proof}

\subsection{Rooted tree isomorphism formulas}

Let $\phi: \G{f_1}{X_1} \rightarrow \G{f_2}{X_2}$ be a homomorphism of functional graph. We consider a point $\alpha \in X_1$ and the sets $P_1=\pred{\alpha}{f_1}{X_1}$ and $P_2=\pred{\phi(\alpha)}{f_2}{X_2}$. When $\phi(P_1)\subseteq P_2$ and the restricted map $\phi|_{P_1}:P_1\rightarrow P_2$ is a bijection, this map determines an isomorphism between the rooted trees $T_1=\tree{\alpha}{f_1}{X_1}$ and $T_2=\tree{\phi(\alpha)}{f_2}{X_2}$ (i.e. a bijection between the vertices preserving directed edges). In this case we say that $\phi|_{P_1}: T_1 \rightarrow T_2$ is a rooted tree isomorphism and the trees $T_1$ and $T_2$ are isomorphic which is denoted by $T_1\simeq T_2$. Sometimes, when the context is clear, we abuse notation and write $T_1=T_2$ when these trees are isomorphic. 

The first result is about the trees attached to the map 
$r_n(\alpha)=\alpha^n$. Since $\Fq^{*}$ and $H$ are closed 
under multiplication we have $r_n(\Fq^{*})\subseteq \Fq^{*}$ 
and $r_n(H)\subseteq H$.

\begin{proposition}\label{PropTreeIsomNmap}
Let $q-1=\nu_0\omega_0$ and $q+1=\nu_1\omega_1$ be the 
$n$-decomposition of $q-1$ and $q+1$, respectively. Let 
$\alpha \in \Fq^{*}$ and $\beta \in H$ be two $r_n$-periodic 
points. Then $\tree{\alpha}{r_n}{\Fq^{*}}=T_{\nu_0(n)}$ and 
$\tree{\beta}{r_n}{H}=T_{\nu_1(n)}$.
\end{proposition}

\begin{proof}
The sets $\Fq^{*}$ and $H$ are multiplicative cyclic groups of order 
$q-1$ and $q+1$, respectively. In general, if $G$ is a multiplicative 
cyclic group of order $m=\nu\omega$ with $\rad$, $\gcd(n,\omega)=1$, 
and $r_n:G\rightarrow G$ is the map given by $r_n(g)=g^n$ we prove 
that $\tree{g_0}{r_n}{G}=T_{\nu(n)}$. Indeed, if $\xi$ is a generator 
of $G$ and $\phi: \Z_m \rightarrow G$ is the map given by 
$\phi(i)=\xi^{i}$, then $r_n\circ \phi (i) = (\xi^{i})^{n} = 
\xi^{ni}=\phi \circ n (i)$ (where $n$ denotes the multiplication-by-$n$ 
map). This implies that $\phi: \G{n}{\Z_m} \rightarrow \G{r_n}{G}$ 
is an isomorphism of functional graphs. Since all the trees attached 
to periodic points in $\G{n}{\Z_m}$ are isomorphic to $T_{\nu(n)}$ 
(Theorem \ref{ThNmapIsom}) the same occurs for the trees attached 
to periodic points in $\G{r_n}{G}$.  \hfill $\qed$
\end{proof}

\begin{proposition}\label{PropTreeIsoOp}
If $n\geq 1$ is an odd integer and $a\in \Fq$, then 
$\tree{a}{T_n}{\Fq}$ and $\tree{-a}{T_n}{\Fq}$ are isomorphic.
\end{proposition}

\begin{proof}
Consider the map $\mbox{op}: \Fq \rightarrow \Fq$ given by $\mbox{op}(x)=-x$. Since $n$ is an odd integer, the Chebyshev polynomial is an odd function and we have $\mbox{op}\circ T_n = T_n \circ \mbox{op}$. Thus $\mbox{op}: \GT{\Fq} \rightarrow \GT{\Fq}$ is an isomorphism of functional graphs and the results follows from Proposition \ref{PropCoveringAndTreeIsom}.  \hfill $\qed$
\end{proof}

\begin{proposition}\label{PropTreeIsoInv}
Let $\alpha \in \Fqt$. Then, $\tree{\alpha}{r_n}{\Fqt}$ and 
$\tree{\alpha^{-1}}{r_n}{\Fqt}$ are isomorphic.
\end{proposition}

\begin{proof}
We consider the isomorphism of functional graphs 
$i: \G{r_n}{\Fqt} \rightarrow \G{r_n}{\Fqt}$ given by 
$i(x)=x^{-1}$ (it is an isomorphism because $i:\Fqt \rightarrow \Fqt$ 
is bijective and  $i\circ r_n = r_n \circ i$). The results follows 
from Proposition \ref{PropCoveringAndTreeIsom}.  \hfill $\qed$
\end{proof}

\begin{proposition}\label{PropTreeIsoEta}
Let $\alpha \in \Fqt$ with $\alpha \neq \pm 1$ and $a=\eta(\alpha)$. 
Then, $\tree{\alpha}{r_n}{\Fqt}$ and $\tree{a}{T_n}{\Fq}$ 
are isomorphic.
\end{proposition}

\begin{proof}
We consider the homomorphism 
$\eta: \G{r_n}{\Fqt} \rightarrow \G{T_n}{\Fq}$ (it is a homomorphism 
because $\eta\circ r_n = T_n \circ \eta$). This homomorphism is in 
fact a $i$-covering because 
$\eta^{-1}(a)=\{\alpha, i(\alpha)=\alpha^{-1}\}$ where 
$\alpha \in \Fqt$ is a root of $x^2-ax+1=0$. We note that 
$\alpha \in \Fqt$ is not $i$-maximal if and only if 
$\alpha=\alpha^{-1}$ since $i$ is a permutation of order $2$; this 
is equivalent to $\alpha=\pm 1$. If $\alpha \neq \pm 1$, then 
$\alpha$ is $i$-maximal and the result follows from Proposition 
\ref{PropCoveringAndTreeIsom}.  \hfill $\qed$
\end{proof}

\section{Splitting the functional graph $\mathcal{G}(T_n/\F_q)$ 
into uniform components} \label{uniformcomponents}

The most simple case of functional graph is when the trees attached 
to the periodic points are isomorphic. In this case describing the 
functional graph is equivalent to describing the cycle decomposition 
of the periodic points and the rooted tree attached to any periodic 
point. We start with a definition.

\begin{definition}
A functional graph $\mathcal{G}(f/X)$ is \emph{uniform} if for 
every pair of periodic points $x,x' \in X$ the trees $\tree{x}{f}{X}$ 
and $\tree{x'}{f}{X}$ are isomorphic.
\end{definition}

In this section we decompose the set $\Fq$ in three $T_n$-invariant 
sets: $R$ (the rational component), $Q$ (the quadratic component) 
and $S$ (the special component), obtaining a decomposition of the 
Chebyshev functional graph 
\begin{equation}\label{EqSplitting}
\GT{\Fq} = \GT{R} \oplus \GT{Q} \oplus \GT{S}. 
\end{equation}
Moreover, we prove that the functional graphs of the right hand side 
are uniform (Proposition \ref{PropUniformity}). We describe each 
component separately.
 
\begin{lemma}\label{LemmaEtaInv}
We have $X \subseteq \Fq$ is $T_n$-invariant if and only if 
$\eta^{-1}(X)$ is $r_n$-invariant.
\end{lemma}

\begin{proof}
($\Rightarrow$) Let $\alpha \in \eta^{-1}(X)$. We have 
$\eta(\alpha)\in X$ and $T_n(\eta(\alpha)) \in X$ (because $X$ is 
forward $T_n$-invariant). Therefore 
$\eta(r_n(\alpha)) =  T_n(\eta(\alpha)) \in X$ and then 
$r_n(\alpha) \in \eta^{-1}(X)$. This proves that $\eta^{-1}(X)$ is 
forward $r_n$-invariant. Now we consider $\beta \in \Fqt$ such 
that $r_n(\beta)=\alpha \in \eta^{-1}(X)$. Then $T_n(\eta(\beta))
= \eta(r_n(\beta)) \in X$. Since $X$ is backward $T_n$-invariant 
$\eta(\beta) \in X$, thus $\beta \in \eta^{-1}(X)$. This proves 
that $\eta^{-1}(X)$ is backward $r_n$-invariant.

($\Leftarrow$) Let $x \in X$. Since $\eta$ is surjective we can 
write $x=\eta(\alpha)$ for some $\alpha \in \Fqt$. We have 
$\alpha \in \eta^{-1}(X)$ and using that $\eta^{-1}(X)$ is forward 
$r_n$-invariant we also have $r_n(\alpha) \in \eta^{-1}(X)$. Thus 
$T_n(x)= T_n(\eta(\alpha)) = \eta(r_n(\alpha)) \in X$. This proves 
that $X$ is forward $T_n$-invariant. Now we consider $y \in \Fq$ 
such that $T_n(y)=x\in X$ and we can write $y=\eta(\beta)$ with 
$\beta \in \Fqt$ since $\eta$ is surjective. We have that 
$T_n(y)=T_n(\eta(\beta))=\eta(r_n(\beta)) \in X$, thus 
$r_n(\beta) \in \eta^{-1}(X)$. Using that $\eta^{-1}(X)$ is 
backward $r_n$-invariant we conclude that $\beta \in \eta^{-1}(X)$. 
Therefore $y=\eta(\beta)\in X$ which proves that $X$ is backward 
$T_n$-invariant.   \hfill $\qed$
\end{proof}

Using the characterizations $\Fq^{*}=\{\alpha \in \Fqt: 
\mbox{ord}(\alpha)\mid q-1\}$ and $H= \{\alpha \in \Fqt: 
\mbox{ord}(\alpha)\mid q+1\}$, we obtain the following 
decomposition of $\Fqt$ into $r_n$-invariant subsets.

\begin{lemma}\label{Lemmarninvariance}
The subsets $\tilde{S} = \{\alpha\in \Fqt: \alpha^{n^k}
= \pm 1 \textrm{ for some }k\geq 0\}$, $\tilde{R} 
= \Fq^{*} \setminus \tilde{S}$ and 
$\tilde{Q} = H \setminus \tilde{S}$ form a partition of 
$\Fqt$ in $r_n$-invariant subsets.
\end{lemma}  

\begin{proof}
Since $(\pm 1)^n \subseteq \{\pm 1\}$, the set $\tilde{S}$ 
is forward $r_n$-invariant. If $\alpha^n \in \tilde{S}$ 
there exists $k\geq 0$ such that $(\alpha^n)^{n^k} = 
\alpha^{n^{k+1}}=\pm 1$. Thus $\alpha \in \tilde{S}$ and 
$\tilde{S}$ is backward $r_n$-invariant. This proves that 
$\tilde{S}$ is $r_n$-invariant.

The proofs of the $r_n$-invariance of $\tilde{R}$ and 
$\tilde{Q}$ are similar. We only prove that $\tilde{R}$ 
is $r_n$-invariant. It is easy to prove that the complement 
of an $r_n$-invariant is $r_n$-invariant and the intersection 
of two $r_n$-invariant sets is also $r_n$-invariant. Since 
$R= \Fq^{*} \cap \tilde{S}^{c}$, it suffices to prove that 
$\Fq^{*}$ is $r_n$-invariant. It is clear that $\Fq^{*}$ is 
forward $r_n$-invariant. To prove that $\Fq^{*}$ is backward 
$r_n$-invariant we use the characterization 
$\Fq^{*}=\{\alpha \in \Fqt: \mbox{ord}(\alpha)\mid q-1\}$. We consider 
$\beta \in \Fqt$ such that $r_n(\beta)=\beta^{n}\in \Fq^{*}$. 
The multiplicative order of $\beta^n$ is given by 
$\mbox{ord}(\beta^n) = \mbox{ord}(\beta)/d$ with 
$d=\gcd(\mbox{ord}(\beta),n)$. In particular 
$\mbox{ord}(\beta)\mid q-1$ (because 
$\mbox{ord}(\beta)\mid \mbox{ord}(\beta^n)$ and 
$\mbox{ord}(\beta^n)\mid q-1$), therefore 
$\beta \in \Fq^{*}$ by the above characterization of 
$\Fq^{*}$.  \hfill $\qed$
\end{proof}

\begin{proposition}\label{PropTninvariance}
Let $R= \eta(\tilde{R}), Q= \eta(\tilde{Q})$ and 
$S= \eta(\tilde{S})$. The sets $R, Q$ and $S$ form a partition 
of $\Fq$ in $T_n$-invariant sets. In particular the decomposition 
of $\GT{\Fq}$ given by (\ref{EqSplitting}) holds.
\end{proposition}

\begin{proof}
It is straightforward to check that $\tilde{R}, \tilde{Q}$ and 
$\tilde{S}$ are $i$-invariant from which we obtain 
$\eta^{-1}(R)= \tilde{R}$, $\eta^{-1}(Q)= \tilde{Q}$ and 
$\eta^{-1}(S)= \tilde{S}$. By Lemma \ref{Lemmarninvariance} 
these sets are $r_n$-invariant, and by Lemma \ref{LemmaEtaInv} 
$R,Q$ and $S$ are $T_n$-invariant.   \hfill $\qed$ 
\end{proof}

We finish this section proving that the functional graphs 
$\GT{R}, \GT{Q}$ and $\GT{S}$ are uniform.

\begin{proposition}\label{PropUniformity}
The functional graphs $\GT{R}, \GT{Q}$ and $\GT{S}$ are uniform. 
Moreover, every tree attached to a $T_n$-periodic point in $\GT{R}$ 
is isomorphic to $T_{\nu_0(n)}$ and every tree attached to a 
$T_n$-periodic point in $\GT{Q}$ is isomorphic to $T_{\nu_1(n)}$.
\end{proposition}

\begin{proof}
The easy case is to prove that $\GT{S}$ is uniform, the other two 
cases are similar and we prove only that $\GT{R}$ is uniform. If 
$n$ or $q$ is even, the only $T_n$-periodic point in $S$ is $2$ 
and there is nothing to prove. If $n$ and $q$ are odd there are 
two $T_n$-periodic points in $S$, $2$ and $-2$, and the uniformity 
of $\GT{S}$ follows from Proposition \ref{PropTreeIsoOp}. 

We denote by $P_{f}$ the set of periodic points with respect to 
$f$ and consider $a\in R \cap P_{f}$. We can write $a=\eta(\alpha)$ 
for some $\alpha \in \tilde{R}$ (in particular $a\in \Fq^{*}$ and 
$a\neq \pm 1$). By Proposition \ref{PropTreeIsoEta}, 
$\tree{a}{T_n}{\Fq}$ and $\tree{\alpha}{r_n}{\Fqt}$ are isomorphic. 
Using that $\Fq^{*}$ is $r_n$-invariant and $a\in \Fq^{*}$ we have 
$\tree{\alpha}{r_n}{\Fqt}= \tree{\alpha}{r_n}{\Fq^{*}}$ and by 
Proposition \ref{PropCoveringPreservePer} (considering the 
$i$-covering $\eta: \G{r_n}{\Fqt} \rightarrow \GT{\Fq}$) we have 
that $\alpha$ is an $r_n$-periodic point. By Proposition 
\ref{PropTreeIsomNmap} we have that $\tree{\alpha}{r_n}{\Fqt}$ 
is isomorphic to $T_{\nu_0(n)}$ and by transitivity 
$\tree{a}{T_n}{\Fq}$ is also isomorphic to $T_{\nu_0(n)}$. 
\hfill $\qed$ 
\end{proof}

\section{The rational and quadratic components}
\label{rationalquadratic}

In this section we describe the functional graphs $\GT{R}$ and $\GT{Q}$.

The following proposition is a simple generalization of Proposition 
2.1 of \cite{Gassert14} for the general $n$ case and is proved in 
a similar way.

\begin{proposition}\label{PropPerAndPPer}
Let $a\in \F_q$, $\alpha \in \Fqt$ such that $a=\alpha+\alpha^{-1}$ and 
$\mbox{ord}(\alpha)=ud$ the $n$-decomposition of the (multiplicative) 
order of $\alpha$. Then $\mbox{per}(a)=\tilde{o}_{d}(n)$ and 
$\mbox{pper}(a)=\min\{k\geq 0: u \mid n^{k}\}$.
\end{proposition}

\begin{proof}
Let $\pi=\mbox{per}(a)$ and $\rho=\mbox{pper}(a)$. Consider the 
following equivalences: 
\begin{displaymath}
\begin{split}
T_{n}^{\pi+\rho}(a)=T_{n}^{\rho}(a) & \Leftrightarrow T_{n^{\pi+\rho}}(a)=T_{n^{\rho}}(a)\\ 
& \Leftrightarrow \alpha^{n^{\pi+\rho}}+\alpha^{-n^{\pi+\rho}}=\alpha^{n^{\rho}}+\alpha^{-n^{\rho}} \\ 
& \Leftrightarrow (\alpha^{n^{\pi+\rho}}-\alpha^{n^{\rho}})(\alpha^{n^{\pi+\rho}}-\alpha^{-n^{\rho}})=0 \\
& \Leftrightarrow \alpha^{n^{\pi+\rho}}=\alpha^{n^{\rho}} \textrm{ or } \alpha^{n^{\pi+\rho}}=\alpha^{-n^{\rho}} \\
& \Leftrightarrow n^{\pi+\rho}\equiv \pm n^{\rho}\!\!\!\!\pmod{ud} \\
& \Leftrightarrow n^{\pi}\equiv \pm 1\!\!\!\! \pmod{d} \textrm{ and } u\mid n^{\rho}.
\end{split}
\end{displaymath}
By minimality, we conclude that $\pi=\tilde{o}_{d}(n)$ and 
$\rho=\min\{k\geq 0 : u\mid n^{k}\}$. \hfill $\qed$
\end{proof}

\begin{corollary}\label{CorPeriodicPoints}
Let $\alpha \in \Fqt$. The point $a=\alpha+\alpha^{-1} \in \F_q$ is 
$T_n$-periodic point if and only if the multiplicative order of 
$\alpha$ (as element of $\F_{q^2}^{*}$) is coprime with $n$. 
\end{corollary}

\begin{proof}
Let $a=\alpha+\alpha^{-1} \in \F_q$ and $\mbox{ord}(\alpha)=ud$ be 
the $n$-decomposition of the (multiplicative) order of $\alpha$. 
We have that $a$ is $T_n$-periodic point if and only if 
$\mbox{pper}(a)=0$ and by Proposition \ref{PropPerAndPPer} this 
happens if and only if $u\mid 1$, that is, if and only if $u=1$ 
and $\gcd(\mo{\alpha},n)=1$. \hfill $\qed$ 
\end{proof}

\begin{corollary}\label{CorPeriodicityAndOrders}
Let $P_{T_n}$ be the set of $T_n$-periodic points, $\alpha \in \Fqt$ 
and $a=\alpha+\alpha^{-1}$.
\begin{enumerate}
\item[1.] $a \in R \cap P_{T_n}$ if and only if 
$\mbox{ord}(\alpha)>2$ and $\mbox{ord}(\alpha)\mid \omega_0$;
\item[2.] $a \in Q \cap P_{T_n}$ if and only if 
$\mbox{ord}(\alpha)>2$ and $\mbox{ord}(\alpha)\mid \omega_1$;
\item[3.] $a \in S \cap P_{T_n}$ if and only if 
$\mbox{ord}(\alpha)\leq 2$ and $\gcd(\mo{\alpha},n)=1$.
\end{enumerate}
\end{corollary}

\begin{proof}
Since $\eta$ is surjective, $\eta(\eta^{-1}(X))=X$ for all $X\subseteq \Fq$ 
(in particular $a\in X$ if and only if $\alpha \in \eta^{-1}(X)$). Denote 
$\tilde{P}_{T_n}:= \eta^{-1}(P_{T_n})$. By Corollary \ref{CorPeriodicPoints}, 
$\tilde{P}=\{\alpha \in \Fqt: \gcd(\mo{\alpha},n)=1\}$. First we prove 
that $\tilde{P}_{T_n} \cap \tilde{S} = \tilde{P}_{T_n}\cap \{+1\}$. Indeed, 
if $\alpha \in \tilde{P}_{T_n} \cap \tilde{S}$, then $\gcd(\mo{\alpha},n)=1$ 
and $\alpha^{n^k}=\pm 1$ for some $k\geq 0$. Thus 
$\mo{\alpha}= \frac{\mo{\alpha}}{\gcd(\mo{\alpha},n^k)} = 
 \mo{\alpha^{n^k}}=\mo{\pm 1}|2$ which implies $\alpha = \pm 1$. This proves 
that $\tilde{P}_{T_n} \cap \tilde{S} \subseteq \tilde{P}_{T_n}\cap \{+1\}$ 
and the other inclusion is clear. We note that this is equivalent to 
$\tilde{P}_{T_n} \cap \tilde{S}^{c} = \tilde{P}_{T_n}\cap \{+1\}^{c}$. 
Now we prove the statements.

\noindent {\it 1.} $a \in R \cap P_{T_n}$ if and only if 
$\alpha \in \tilde{R} \cap \tilde{P}_{T_n} = 
\Fq^{*} \cap \tilde{S}^{c} \cap \tilde{P}_{T_n} = 
\tilde{P}_{T_n}\cap \Fq^{*}\cap \{\pm 1\}^{c} = 
\{\alpha \in \Fqt\!\! : \gcd(\mo{\alpha},n)\!=
\!1,\mo{\alpha}\!\!\mid\!\! q-1, \alpha \neq \pm 1\} = 
\{\alpha \in \Fqt\!\!: \mo{\alpha}\!\!\mid\!\! \omega_0,$ $\mo{\alpha}>2\}$.

\noindent {\it 2.} This part is similar to {\it 1.}; here we use 
$\alpha \in H$ if and only if $\mo{\alpha}\mid q+1$.

\noindent {\it 3.} $a \in S \cap P_{T_n}$ if and only if 
$\alpha \in \tilde{S} \cap \tilde{P}_{T_n} = 
\tilde{P}_{T_n} \cap \{\pm 1\}= \{\alpha \in \Fqt: \gcd(\mo{\alpha},n)=1, 
\mo{\alpha}\leq 2\}$.  \hfill $\qed$ 
\end{proof}

Next we obtain an isomorphism formula for the rational component 
and the quadratic component of $\GT{\Fq}$.

\begin{theorem}\label{ThRationalQuadratic} 
Let $q-1=\nu_{0}\omega_{0}$ and $q+1=\nu_{1}\omega_{1}$ be their 
$n$-decompositions. The rational component of the Chebyshev's graph 
$\mathcal{G}(T_n/\F_q)$ is given by:
$$\GT{R}= \bigoplus_{\substack{d\mid \omega_0\\  d> 2}} 
  \frac{\varphi(d)}{2\tilde{o}_{d}(n)}\times\mbox{Cyc}
  \left(\tilde{o}_{d}(n),T_{\nu_0(n)} \right);$$  
the quadratic component is given by 
$$\GT{Q}= \bigoplus_{\substack{d\mid \omega_1\\  d> 2}} 
  \frac{\varphi(d)}{2\tilde{o}_{d}(n)}\times\mbox{Cyc}
  \left(\tilde{o}_{d}(n),T_{\nu_1(n)} \right). $$
\end{theorem}

\begin{proof}
We only prove the statement for the rational component since 
the proof for the quadratic component is similar. Let $P_{T_n}$ 
be the set of $T_n$-periodic points and 
$R_d=\{\alpha+\alpha^{-1}: \alpha \in \Fqt, \mo{\alpha}=d\}$. 
By Corollary \ref{CorPeriodicityAndOrders}, $R\cap P_{T_n}$ is 
the disjoint union of $R_d$ with $d\mid \omega_{0}, d>2$. If 
$\mo{\alpha}= d\mid \omega_{0}$ we have that $\gcd(d,n)=1$ and 
$\mo{\alpha^n}=\mo{\alpha}/\gcd(\mo{\alpha},n)=\mo{\alpha}$. 
Then we have the following decomposition $\GT{R\cap P_{T_n}} =
\bigoplus_{\substack{d\mid \omega_0\\  d> 2}} \GT{R_d}.$ 
By Proposition \ref{PropPerAndPPer}, every point in $\GT{R_d}$ 
belongs to a cycle of length $\tilde{o}_{d}(n)$. Thus,
\begin{equation}\label{EqGraphRper}
\GT{R\cap P_{T_n}} =\bigoplus_{\substack{d\mid \omega_0\\  d> 2}} 
  \frac{\#R_d}{\tilde{o}_{d}(n)} \times \mbox{Cyc}
  \left(\tilde{o}_{d}(n),\bullet \right).
\end{equation}
For each $d\mid \omega_{0}, d>2$, we consider the set 
$\tilde{R}_{d}=\{\alpha \in \Fqt: \mbox{ord}(\alpha)=d\}$. 
By a standard counting argument $\#\tilde{R}_{d}=\varphi(d)$ 
and using that the restriction of $\eta$ to $\tilde{R}$ is a 
$2$-to-$1$ map from $\tilde{R}$ onto $R$ we obtain 
$\#R = \# \tilde{R}/2= \varphi(d)/2$. Substituting this 
expression into Equation (\ref{EqGraphRper}) and using the 
uniformity of $\GT{R}$ (Proposition \ref{PropUniformity}) we obtain 
$\GT{R}= \bigoplus_{\substack{d\mid \omega_0\\  d> 2}} 
   \frac{\varphi(d)}{2\tilde{o}_{d}(n)}\times\mbox{Cyc}
   \left(\tilde{o}_{d}(n),T_{\nu_0(n)} \right)$. \hfill $\qed$ 
\end{proof}

\section{The special component of $\mathcal{G}(T_n/\F_q)$}
\label{special}

In this section we describe the special component of the Chebyshev 
functional graph $\GT{S}$ where 
$S=\{a \in \F_q: T_n(a)^{(k)}=\pm 2,\textrm{ for some }k\geq 0\}$. 
If $n$ and $q$ are odd, $T_n(-2)=-2$ and $T_n(2)=2$ then the only 
periodic points of $T_n$ in $S$ are $2$ and $-2$. In this case 
the trees attached to the fixed points $2$ and $-2$ are isomorphic 
(Proposition \ref{PropUniformity}). If either $n$ is even or $q$ 
is even, $T_n(-2)=2=T_n(2)$ and the only periodic point of $T_n$ 
in $S$ is $2$ (if $q$ is even this is true because $2=-2$). 
The next proposition summarizes the above discussion.

\begin{proposition}\label{PropSpecialComp1}
Let $\mathcal{T}=\tree{2}{T_n}{\Fq}$ be the rooted tree attached to the 
fixed point $2$ for the Chebyshev polynomial $T_n$ restricted to the 
set $S=\{a \in \F_q: T_n(a)^{(k)}=\pm 2,\textrm{ for some }k\geq 0\}$. 
Then 
$$\GT{S} = \left\{\begin{array}{ll}
2\times \mbox{Cyc}(1,T) & \textrm{if $n$ is odd and $q$ is odd;}\\
\mbox{Cyc}(1,T) & \textrm{otherwise.}   
\end{array} \right.$$
\end{proposition}

We remark that $\tree{2}{T_n}{S}=\tree{2}{T_n}{\Fq}$, which is a
consequence of $S$ being $T_n$-invariant (Proposition 
\ref{PropTninvariance}). By Proposition \ref{PropSpecialComp1}, 
to describe the special component it suffices to describe the 
tree $\mathcal{T}=\tree{2}{T_n}{\Fq}$. If $q-1=\nu_0\omega_0$ 
and $q+1=\nu_1\omega_1$ is the $n$-decomposition of $q-1$ and 
$q+1$, respectively, the rooted trees attached to the periodic 
points are isomorphic to $T_{\nu_0(n)}$ in the rational component 
and isomorphic to $T_{\nu_1(n)}$ in the quadratic component 
(Proposition \ref{PropUniformity}). In the case of the special component 
the situation is different, the tree $\mathcal{T}=\tree{2}{T_n}{\Fq}$ 
is not isomorphic to a tree associated to a $\nu$-series (that is, 
the trees associated to the multiplication by $n$ map over $\Z_{m}$ 
for some $m\in\Z^{+}$). However we show in this section that the tree 
$\mathcal{T}$ can be expressed as a ``mean'' of the trees $T_{\nu_{0}(n)}$ 
and $T_{\nu_1{0}(n)}$. In the first part of this section we define 
the bisection of trees together some of their main properties. In the 
second part we deduce an isomorphism formula for the special component 
of the Chebyshev graph.

\subsection{Bisection of rooted trees}

We start by defining the sum of rooted trees.

\begin{definition}
Let $T=\langle T_1\oplus T_2\oplus \cdots \oplus T_r \rangle$ and 
$T'=\langle T_1'\oplus T_2'\oplus \cdots \oplus T_s' \rangle$ be two 
rooted trees. We define their sum as 
$T+T'=\langle T_1\oplus T_2\oplus \cdots \oplus T_r \oplus 
T_1'\oplus T_2'\oplus \cdots \oplus T_s' \rangle$. 
\end{definition}

We remark that the tree consisting of a unique node 
$T=\bullet=\langle \emptyset \rangle$ is the neutral element of the 
sum. The tree $T-T'$ denotes a tree such that $T=T'+(T-T')$ in case 
this tree exists (if it exists, it is unique up to isomorphism). We 
note that $(T_1+T_2)-T'$ is defined if and only if $T_i-T'$ is defined 
for some $i=1,2$. If $T_1-T'$ is defined then $(T_1+T_2)-T'=(T_1-T')+T_2$ 
and if $T_2-T'$ is defined then $(T_1+T_2)-T'=T_1+(T_2-T')$. Therefore 
when $(T_1+T_2)-T'$ is defined we can write this tree as $T_1+T_2-T'$ 
without ambiguity.

A forest is a graph that can be expressed as a disjoint union of 
rooted trees. A tree $T$ is \emph{even} if it can be expressed as 
$T=\langle 2\times F\rangle$ for some forest $F$ and it is 
\emph{quasi-even} if it can be expressed as 
$T=\langle 2\times F \oplus T' \rangle$ for some forest $F$ and some 
even tree $T'$ (i.e. $T'=\langle 2 \times F' \rangle$ for some forest 
$F'$). In particular the tree $T=\bullet$ is even because 
$T=\langle 2\times \emptyset \rangle$. For these classes of trees 
we define the bisection as follows.

\begin{definition}
If $T=\langle 2\times F \rangle$ is an even tree, its \emph{bisection} 
is the tree $\frac{1}{2}T = \langle F \rangle$. If 
$T=\langle 2 \times F \oplus \langle 2 \times F' \rangle \rangle$ is 
a quasi-even tree its \emph{bisection} is defined as the tree 
$\frac{1}{2}T = \langle F \oplus \langle F' \rangle \rangle$.
\end{definition}

\begin{example}
The tree associated with the $v$-series $18(30)=(6,3)$ is given by 
$T_{(6,3)}=\langle 2\times T \oplus 3\times T' \rangle$ where 
$T=\langle 6 \times \bullet \rangle$ and $T'=\bullet$. Thus $T_{(6,3)}$ 
is quasi-even since it can be written as 
$T_{(6,3)}=\langle 2\times F \oplus T' \rangle$ with $F=T\oplus T'$ 
and $T'=\langle 2\times \emptyset \rangle$ is even. The bisection of 
this tree is given by 
$\frac{1}{2}T_{(6,3)} = \langle F \oplus \langle \emptyset 
\rangle\rangle = \langle T\oplus 2\times T' \rangle $.
\end{example}

Even and quasi-even trees are very restricted classes of trees, however 
they contain all trees associated with $\nu$-series as stated in the 
following proposition.

\begin{proposition}\label{PropTreeVseriesEven}
If $T_{\nu(n)}$ is the tree associated with 
$\nu(n)=(\nu_1,\ldots, \nu_D)$, then $T_{\nu(n)}$ is even when 
$\nu$ is odd and quasi-even when $\nu$ is $even$.
\end{proposition}

\begin{proof}
By Equation (\ref{TreeAssociatedEq}) we have  $T_{\nu(n)} = 
\langle (\nu_D-1) \times T^{D-1} \oplus \bigoplus_{i=1}^{D-1}
(\nu_{i}-\nu_{i+1})\times T^{i-1} \rangle$, where the $T_{i}$ 
are pairwise non-isomorphic rooted trees. When $\nu$ is odd, 
$\nu_{i}$ is odd for $1\leq i \leq D$. Then, $\nu_D-1$ and 
$\nu_{i}-\nu_{i+1}$ are even for $1\leq i \leq D-1$ and the 
tree $T_{\nu(n)}$ is even. When $\nu$ is even, we have that 
$\nu_1,\ldots,\nu_{k}$ are even and $\nu_{k+1},\ldots,\nu_{D}$ 
are odd for some $k$, $1\leq k \leq D$. If $k=D$, then 
$\nu_{D}-1$ is odd and $\nu_{i}-\nu_{i+1}$ are even for 
$1\leq i \leq D-1$ and the tree $T_{\nu(n)}$ is quasi-even. 
If $k<D$, then $\nu_{D}-1$ and $\nu_{i}-\nu_{i+1}$ are even 
for $1\leq i \leq k-1$ and $k+1\leq i \leq D$, and 
$\nu_{k}-\nu_{k+1}$ is odd. Thus, $T_{\nu(n)}$ is also quasi-even.
\hfill $\qed$ 
\end{proof}

We note that the if $T_1$ and $T_2$ are rooted trees, then 
$|T_1+T_2|=|T_1|+|T_2|-1$ where, as usual, $|T|$ denotes the 
number of nodes of $T$. The next proposition establishes a 
relation between $|T|$ and $|\frac{1}{2}T|$.
\begin{proposition}\label{PropBisectionCardinality}
Let $T$ be a rooted tree with $|T|=N$ nodes. We have 
$$\left|1/2 \cdot T\right|= \left\{ \begin{array}{ll}
\frac{N+1}{2} & \textrm{if $T$ is even;}\\
\frac{N+2}{2} & \textrm{if $T$ is quasi-even.}
\end{array}  \right.$$
\end{proposition}

\begin{proof}
If $T$ is even, there is a forest $S$ with $s$ nodes such that 
$T=\langle 2 \times S \rangle$. We have $N=|T|=1+2s$ from which 
we obtain $s=\frac{N-1}{2}$. Since $\frac{1}{2}T = \langle S \rangle$, 
$|\frac{1}{2}T|=s+1=\frac{N-1}{2}+1= \frac{N+1}{2}$. If $T$ is 
quasi-even, there is a pair of forests $S$ and $R$ with $s$ and 
$r$ nodes, respectively, such that 
$T=\langle 2 \times S \oplus \langle 2\times R \rangle \rangle$. 
We have $N=|T|=1+2s+1+2r = 2(r+s+1)$ from which we obtain 
$r+s+1= \frac{N}{2}$. Since 
$\frac{1}{2}T = \langle S \oplus \langle R \rangle \rangle$, 
$|\frac{1}{2}T|=1+s+1+r=1+\frac{N}{2} = \frac{N+2}{2}$. 
\hfill $\qed$ 
\end{proof}

\subsection{The tree $\tree{2}{T_n}{\Fq}$}

The next theorem describe the rooted tree attached to the fixed
point $2$ for the Chebyshev polynomial $T_n:\Fq \rightarrow \Fq$.
We require the following lemma.

\begin{lemma}\label{LemmaSpecial}
Let $n>1$ be an even integer, $\Fq$ be an odd characteristic finite field 
and $H$ be the multiplicative subgroup of $\F_{q^2}^{*}$ with order $q+1$.
\begin{itemize}
\item[(i)] If $q\equiv 3 \pmod{4}$, the equation $x^n=-1$ has no solution 
in $\Fq^{*}$.
\item[(ii)] If $q\equiv 1 \pmod{4}$, the equation $x^n=-1$ has no solution 
in $H$.
\end{itemize}
\end{lemma}

\begin{proof}
Let $\alpha \in \Fqt$ be a solution of $x^n=-1$. From the relations 
$\mo{\alpha^n} = \mo{\alpha}/\mbox{gcd}(\mo{\alpha},n)$ and $\mo{-1}=2$, 
we conclude that if $n$ is even, then $4\mid\mo{\alpha}$. By Lagrange 
theorem, $\alpha \in \Fq^{*}$ implies $4\mid q-1$ and 
$q\not\equiv 3\pmod{4}$; and $\alpha \in H$ implies $4\mid q-3$ and 
$q\not\equiv 1 \pmod{4}$. \hfill $\qed$ 
\end{proof}

\begin{theorem}\label{ThBinComponent}
Let $q-1=\nu_{0}\omega_{0}$ and $q+1=\nu_{1}\omega_{1}$ be their $n$-decompositions. The rooted tree associated with the fixed point $2$ is described as follows:
$$\tree{2}{T_n}{\F_q} = \left\{ \begin{array}{ll} 1/2\cdot T_{\nu_0(n)} + 1/2 \cdot T_{\nu_1(n)} & \textrm{ if $n$ is odd or $q$ is even;} \\
1/2\cdot T_{\nu_0(n)} + 1/2\cdot T_{\nu_1(n)} - \langle \bullet \rangle & \textrm{ if $n$ is even and $q$ is odd}.
\end{array} \right.    $$
\end{theorem}

\begin{proof}
The isomorphism formula is obtained after relating $\tree{2}{T_n}{\Fq}$ 
and $\tree{1}{r_n}{\Fqt}$. First we consider the case when $n$ is odd 
or $q$ is even. In this case $r_n(-1)=-1$ or $-1=1$, in both cases we 
have that the predecessors of $1$ in $\tree{1}{r_n}{\Fqt}$ are in 
$\Fq^{*}$ or in $H$ (but not in both). Since the sets $\Fq^{*}$ and 
$H$ are backward $r_n$-invariant (Lemma \ref{Lemmarninvariance}), we have 
$$\tree{1}{r_n}{\Fqt}= \tree{1}{r_n}{\Fq^{*}} + \tree{1}{r_n}{H}
  = T_{\nu_0(n)} + T_{\nu_1(n)},$$
where in the last equality we use Proposition \ref{PropTreeIsomNmap}. Now,
we write $r_n^{-1}(1)\cap \Fq^{*} = \{\alpha_1, \ldots, \alpha_{2s},1\}$ 
with $\alpha_{s+i}=\alpha_{i}^{-1}$, $\alpha_i \neq \pm 1$, for all 
$i: 1\leq i \leq s$ and $r_n^{-1}(1)\cap H = \{\beta_1,\ldots,\beta_{2t},1\}$ 
with $\beta_{t+j}=\beta_{j}^{-1}$, $\beta_j\neq \pm 1$, for all 
$j: 1\leq j \leq t$. Denote by 
$\tilde{T}(\alpha_i):=\tree{\alpha_i}{r_n}{\Fq^{*}}$ for $1\leq i \leq 2s$ 
and $\tilde{T}(\beta_j):=\tree{\beta_j}{r_n}{H}$ for $1\leq j \leq 2t$. 
Using Proposition \ref{PropTreeIsoInv} we have that $ T_{\nu_0(n)} = 
\tree{1}{r_n}{\Fq^*}  = \langle \tilde{T}(\alpha_1) \oplus \cdots \oplus 
\tilde{T}(\alpha_{2s})\rangle = \left\langle 2\! \times\! \left( \tilde{T}
(\alpha_1) \oplus \cdots \oplus  \tilde{T}(\alpha_{s}) \right) \right\rangle, $
from which we obtain $$ 1/2 \cdot T_{\nu_0(n)} = \langle \tilde{T}(\alpha_1) 
\oplus \cdots \oplus  \tilde{T}(\alpha_{s}) \rangle. $$ In the same 
way we obtain $$ 1/2 \cdot T_{\nu_1(n)} = \langle \tilde{T}(\beta_1) 
\oplus \cdots \oplus  \tilde{T}(\beta_{t}) \rangle. $$
Let $a_i=\eta(\alpha_i), T(a_i)=\tree{a_i}{T_n}{\Fq},
b_j=\eta(\alpha_j)$ and $T(b_j)=\tree{b_j}{T_n}{\Fq}$ for $1\leq i\leq s$, 
$1\leq j \leq t$. We have $T_n^{-1}(2)=\{a_1,\ldots, a_s,b_1,\ldots,b_t,2\}$ 
and 
\begin{align*}
\tree{2}{T_n}{\Fq} &= \langle T(a_1)\oplus \cdots \oplus T(a_s) \oplus 
T(b_1)\oplus \cdots \oplus T(b_t) \rangle \\
&= \langle \tilde{T}(\alpha_1)\oplus \cdots \oplus \tilde{T}(\alpha_s) 
\oplus \tilde{T}(\beta_1)\oplus \cdots \oplus \tilde{T}(\beta_t) \rangle 
\quad \textrm{(by Prop.~\ref{PropTreeIsoEta})} \\
&= \langle \tilde{T}(\alpha_1)\oplus \cdots \oplus \tilde{T}(\alpha_s)\rangle 
+ \langle  \tilde{T}(\beta_1)\oplus \cdots \oplus \tilde{T}(\beta_t) \rangle \\
&= 1/2 \cdot  T_{\nu_0(n)} +  1/2 \cdot  T_{\nu_1(n)} .
\end{align*}
Now we consider the case when $n$ is even and $q$ is odd. Here we can 
write $r_n^{-1}(1)\cap \Fq^{*} = \{\alpha_1, \ldots, \alpha_{2s},-1,1\}$ 
with $\alpha_{s+i}=\alpha_{i}^{-1}$, $\alpha_i \neq \pm 1$, for all 
$i: 1\leq i \leq s$, $r_n^{-1}(1)\cap H = \{\beta_1,\ldots,\beta_{2t},-1,1\}$ 
with $\beta_{t+j}=\beta_{j}^{-1}$, $\beta_j\neq \pm 1$, for all 
$j: 1\leq j \leq t$ and $r_n^{-1}(-1)=\{\gamma_1,\ldots,\gamma_{2r} \}$ 
with $\gamma_{r+k}=\gamma_{k}^{-1}$, $\gamma_k \neq \pm 1$, for all 
$k: 1\leq k \leq r$.

Denote by $\tilde{T}(\alpha_i):=\tree{\alpha_i}{r_n}{\Fq^{*}}$ for 
$1\leq i \leq 2s$, $\tilde{T}(\beta_j):=\tree{\beta_j}{r_n}{H}$ for 
$1\leq j \leq 2t$, $\tilde{T}(\gamma_k):=\tree{\gamma_k}{r_n}{\Fqt}$ 
for $1\leq k \leq 2r$  and $\tilde{T}(-1):=\tree{-1}{r_n}{\Fqt}$. In 
this case we have, by Proposition \ref{PropTreeIsoInv}, 
$\tilde{T}(-1)= \langle \tilde{T}(\gamma_{1})\oplus 
\cdots \oplus \tilde{T}(\gamma_{2r}) \rangle = \langle 2\times 
(\tilde{T}(\gamma_{1})\oplus \cdots \oplus \tilde{T}(\gamma_{r})) 
\rangle$, thus 
\begin{equation}\label{E}
1/2 \cdot \tilde{T}(-1) = \langle \tilde{T}(\gamma_{1})\oplus \cdots 
\oplus \tilde{T}(\gamma_{r}) \rangle.
\end{equation}
Let $a_i=\eta(\alpha_i), T(a_i)=\tree{a_i}{T_n}{\Fq},
b_j=\eta(\alpha_j)$, $T(b_j)=\tree{b_j}{T_n}{\Fq}$, $c_k=\eta(\gamma_k)$, 
$T(c_k)=\tree{c_k}{T_n}{\Fq}$ for $1\leq i\leq s$, $1\leq j \leq t$, 
$1\leq k\leq r$ and $T(-2)=\tree{-2}{T_n}{\Fq}$. We have 
$T_n^{-1}(2)=\{a_1,\ldots, a_s,b_1,\ldots,b_t, -2, 2\}$, 
$T_n^{-1}(-2)=\{c_1,\ldots, c_r \}$. By Proposition \ref{PropTreeIsoInv} 
and Equation (\ref{E}) we have $T(-2)=\langle T(c_1)\oplus \cdots 
\oplus T(c_r) \rangle = \langle \tilde{T}(\gamma_1)\oplus \cdots 
\oplus \tilde{T}(\gamma_r) \rangle = 1/2 \cdot \tilde{T}(-1)$, thus
$$\tree{2}{T_n}{\Fq} = \langle T(a_1)\oplus \cdots \oplus T(a_s) 
\oplus T(b_1) \oplus \cdots \oplus T(b_t) \oplus T(-2) \rangle$$  
\begin{equation}\label{E1}
= \langle \tilde{T}(\alpha_1)\oplus \cdots \oplus \tilde{T}(\alpha_s) \oplus \tilde{T}(\beta_1) \oplus \cdots \oplus \tilde{T}(\beta_t) \oplus  1/2 \cdot \tilde{T}(-1) \rangle.
\end{equation}
Now we consider two subcases: $q\equiv 1\! \pmod{4}$ and $q\equiv 3\! 
\pmod{4}$. First we consider the subcase $q\equiv 1 \! \pmod{4}$. By 
Lemma \ref{LemmaSpecial} we have $r_n^{-1}(-1)\cap H = \emptyset$ and 
$r_n^{-1}(-1)\subseteq \Fq^{*}$. Thus 
$\tilde{T}(-1)= \tree{-1}{T_n}{\Fq^{*}}$ and we have, by Propositions 
\ref{PropTreeIsomNmap} and \ref{PropTreeIsoInv}, $T_{\nu_0(n)} = 
\tree{1}{r_n}{\Fq^{*}} = \langle \tilde{T}(\alpha_1)\oplus \cdots 
\oplus \tilde{T}(\alpha_{2s}) \oplus \tilde{T}(-1) \rangle = 
\langle 2\times (\tilde{T}(\alpha_1)\oplus \cdots \oplus 
\tilde{T}(\alpha_{s})) \oplus \tilde{T}(-1) \rangle$. Therefore
\begin{equation}\label{E2}
1/2 \cdot T_{\nu_0(n)} = \langle \tilde{T}(\alpha_1)\oplus \cdots \oplus 
\tilde{T}(\alpha_s) \oplus  1/2 \cdot \tilde{T}(-1) \rangle.
\end{equation}
Since $r_n^{-1}(-1)\cap H = \emptyset$, we have $T_{\nu_1(n)} = \tree{1}{r_n}{H} = \langle \tilde{T}(\beta_1) \oplus \cdots \oplus \tilde{T}(\beta_{2t}) \oplus \bullet \rangle = \langle 2\times (\tilde{T}(\beta_1) \oplus \cdots \oplus \tilde{T}(\beta_{t})) \oplus \bullet \rangle$ and $1/2 \cdot T_{\nu_1(n)} = \langle \tilde{T}(\beta_1)\oplus \cdots \oplus \tilde{T}(\beta_t) \oplus \bullet \rangle = \langle \tilde{T}(\beta_1)\oplus \cdots \oplus \tilde{T}(\beta_t)\rangle + \langle  \bullet \rangle$; from which we obtain
\begin{equation}\label{E3}
1/2 \cdot T_{\nu_1(n)} - \langle \bullet \rangle = \langle \tilde{T}(\beta_1)\oplus \cdots \oplus \tilde{T}(\beta_t)\rangle.
\end{equation}
Substituting Equations (\ref{E2}) and (\ref{E3}) in Equation 
(\ref{E1}) we have
\begin{align*}
\tree{2}{T_n}{\Fq} &=  \langle \tilde{T}(\alpha_1)\oplus \cdots \oplus \tilde{T}(\alpha_s) \oplus \tilde{T}(\beta_1) \oplus \cdots \oplus \tilde{T}(\beta_t) \oplus  1/2 \cdot \tilde{T}(-1) \rangle \\ 
&= \langle \tilde{T}(\alpha_1)\oplus \cdots \oplus \tilde{T}(\alpha_s) \oplus  1/2 \cdot \tilde{T}(-1) \rangle +   \langle \tilde{T}(\beta_1)\oplus \cdots \oplus \tilde{T}(\beta_t)\rangle \\ &= 1/2 \cdot T_{\nu_0(n)} + 1/2 \cdot T_{\nu_1(n)} - \langle \bullet \rangle.
\end{align*}
The proof of the subcase $q\equiv 3 \pmod{4}$ is similar. In this case applying Lemma \ref{LemmaSpecial} we obtain $\tilde{T}(-1)= \tree{-1}{T_n}{H}$ and using the same arguments used for the subcase $q\equiv 1 \pmod{4}$ we obtain 
\begin{equation}\label{F2}
1/2 \cdot T_{\nu_1(n)} = 1/2 \cdot \tree{1}{r_n}{H} = \langle \tilde{T}(\beta_1)\oplus \cdots \oplus \tilde{T}(\beta_t) \oplus  1/2 \cdot \tilde{T}(-1) \rangle
\end{equation}
and 
\begin{equation}\label{F3}
1/2 \cdot T_{\nu_0(n)} - \langle \bullet \rangle = \langle \tilde{T}(\alpha_1)\oplus \cdots \oplus \tilde{T}(\alpha_s)\rangle.
\end{equation}
Using Equations (\ref{E1}), (\ref{F2}) and (\ref{F3}) we have 
$\tree{2}{T_n}{\Fq}= 1/2 \cdot T_{\nu_0(n)} + 1/2 \cdot T_{\nu_1(n)} 
- \langle \bullet \rangle$. 
\hfill $\qed$ 
\end{proof}

\section{Structure theorem for Chebyshev polynomial and consequences}
\label{structural}

\subsection{Isomorphism formula for $\GT{\Fq}$}

We summarize all the information in the following main theorem of this 
paper, which follows from Theorems \ref{ThRationalQuadratic} and 
\ref{ThBinComponent} and Proposition \ref{PropSpecialComp1}.

\begin{theorem}\label{ThMain}
Let $q-1=\nu_0\omega_0$ and $q+1=\nu_1\omega_1$ be the $n$-decomposition 
of $q-1$ and $q+1$, respectively. The Chebyshev graph admits a decomposition 
of the form $\GT{\Fq}=\mathcal{G}^{R} \oplus \mathcal{G}^{Q} \oplus 
\mathcal{G}^{S}$ where the rational component $\mathcal{G}^{R}$ is given by
$$\mathcal{G}^{R} =  \bigoplus_{\substack{d\mid \omega_0\\  d> 2}} 
  \frac{\varphi(d)}{2\tilde{o}_{d}(n)}\times\mbox{Cyc}\left(\tilde{o}_{d}(n),
  T_{\nu_0(n)} \right);$$ 
the quadratic component $\mathcal{G}^{Q}$ is given by 
$$\mathcal{G}^{Q}= \bigoplus_{\substack{d\mid \omega_1\\  d> 2}} 
  \frac{\varphi(d)}{2\tilde{o}_{d}(n)}\times\mbox{Cyc}\left(\tilde{o}_{d}(n),
  T_{\nu_1(n)} \right);$$ 
and the special component $\mathcal{G}^{S}$ is given by 
$$\mathcal{G}^{S}= \left\{ \begin{array}{ll} \mbox{Cyc}(1,1/2\cdot T_{\nu_0(n)} 
  + 1/2\cdot T_{\nu_1(n)} - \langle \bullet \rangle ) 
  & \textrm{if $n$ is even and $q$ is odd;}\\
  2\times \mbox{Cyc}(1,1/2\cdot T_{\nu_0(n)} + 1/2\cdot T_{\nu_1(n)}) 
  & \textrm{if $n$ is odd and $q$ is odd;}\\
  \mbox{Cyc}(1,1/2\cdot T_{\nu_0(n)} + 1/2\cdot T_{\nu_1(n)}) 
  & \textrm{if $q$ is even.}\\
\end{array}   \right.  $$ 
\end{theorem}

\subsection{Examples}

We provide a series of examples showing our main result.

\begin{example}
We consider the Chebyshev polynomial $T_{30}$ over $\F_{19}$ (see 
Figure \ref{Figcheby30q19q23}a). We have $19-1=18=\nu_0\omega_0$ 
with $\nu_0=18, \omega_0=1$ and 
$19+1=20= \nu_1\omega_1$ with $\nu_1=20,\omega_1=1$. Since 
$\omega_0,\omega_1\leq 2$ both the rational and the quadratic 
components of $\mathcal{G}(T_{30}/\F_{19})$ are empty. We calculate 
the $\nu$-series $18(30)=(6,3)$ and $20(30)=(10,2)$ obtaining
$$\mathcal{G}(T_{30}/\F_{19})=\mbox{Cyc}\left(1,\frac{1}{2} T_{(6,3)}
 +\frac{1}{2} T_{(10,2)} - \langle \bullet \rangle \right).$$ 
Thus, the graph $\mathcal{G}(T_{30}/\F_{19})$ consist of a 
loop corresponding to the fix point $2$ and a tree 
$T= \frac{1}{2} T_{(6,3)}+ \frac{1}{2} T_{(10,2)} - \langle 
\bullet \rangle$ attached to this point; see Figure \ref{FigCheby30q19}.
\end{example}

\begin{figure}[h]
\centering
\includegraphics[width=0.9\textwidth]{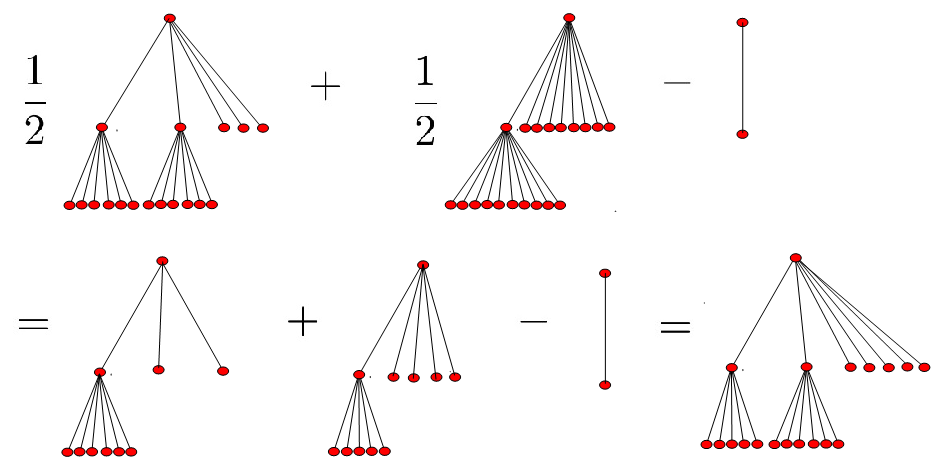}
\caption{Construction of the tree $T= \frac{1}{2} T_{(6,3)}+ 
\frac{1}{2} T_{(10,2)} - \langle \bullet \rangle$.}
\label{FigCheby30q19}
\end{figure}

\begin{example}
Now we consider again the Chebyshev polynomial $T_{30}$ but this time over $\F_{23}$ (see Figure \ref{Figcheby30q19q23}b). We 
have $23-1=22=\nu_0\omega_0$ with $\nu_0=2, \omega_0=11$ and 
$23+1=24= \nu_1\omega_1$ with $\nu_1=24,\omega_1=1$. In this case the quadratic component of $\mathcal{G}(T_{30}/\F_{23})$ is empty and the rational component is $\frac{\varphi(11)}{2\tilde{o}_{11}(30)}\times\mbox{Cyc}\left(\tilde{o}_{11}(30),T_{2(30)}\right)$. Since $\varphi(11)=10$, $\tilde{o}_{11}(30)=5$ and $T_{2(30)}=T_{(2)}=\langle \bullet \rangle$, it is given by $ \mbox{Cyc}\left(5,\langle \bullet \rangle \right)$. We calculate the $\nu$-series $24(30)=(6,2,2)$. Then, the Chebyshev's graph of $T_{30}$ over $\F_{23}$ is given by:
$$\mathcal{G}(T_{30}/\F_{23})=\mbox{Cyc}\left(5,\langle \bullet \rangle \right)\oplus \mbox{Cyc}\left(1,\frac{1}{2} T_{(2)}
 +\frac{1}{2} T_{(6,2,2)} - \langle \bullet \rangle \right).$$ 
We have $\frac{1}{2}T_{(2)}=\frac{1}{2}\langle \langle 2\times \emptyset\rangle\rangle= \langle \langle \emptyset\rangle\rangle= \langle \bullet \rangle$ (i.e. $T_{(2)}$ is invariant under bisection), and after simplifying we obtain $\mathcal{G}(T_{30}/\F_{23})=\mbox{Cyc}\left(5,\langle \bullet \rangle \right)\oplus \mbox{Cyc}\left(1,\frac{1}{2} T_{(6,2,2)} \right)$.
To obtain a more explicit formula we calculate the bisection of $T_{(6,2,2)}$. Using the recursive formula (\ref{TreeAssociatedEq}), we obtain $T_{(6,2,2)}=\langle 4\times \bullet \oplus T \rangle$ where $T=\langle 4\times \bullet \oplus 2\times \langle 6\times \bullet \rangle \rangle$, then $T_{(6,2,2)}$ is quasi-even and $T$ is even. Since $\frac{1}{2}T=\langle 2 \times \bullet \oplus \langle 6\times \bullet \rangle \rangle$, we have $\frac{1}{2}T_{(6,2,2)}=\langle 2\times \bullet \oplus \frac{1}{2}T\rangle = \langle 2\times \bullet \oplus \langle 2 \times \bullet \oplus \langle 6\times \bullet \rangle \rangle$.
\end{example}

\begin{example}
We consider again the Chebyshev polynomial $T_{30}$, this time over 
the reasonably large finite field $\F_{739}$ where the symmetries 
can be better appreciated; see Figure \ref{FigCheby30q739}. We 
calculate the $30$-decomposition of $738=18\cdot 41$ ($\nu_0=18, 
\omega_0=41$) and $740= 120\cdot 37$ ($\nu_1=20,\omega_1=37$). Since 
$\varphi(41)=40$, $\tilde{o}_{41}(30)=20$, $\varphi(37)=36$, 
$\tilde{o}_{37}(30)=9$, the rational component $\mathcal{G}^{R}$ 
and the quadratic component $\mathcal{G}^{Q}$ are given by 
$\mathcal{G}^{R}= \mbox{Cyc}(20,T_{18(30)})$ and 
$\mathcal{G}^{Q}=2 \times \mbox{Cyc}(9, T_{20(30)})$. We have 
$18(30)=(6,3)$ and $20(30)=(10,2)$. Thus the special component is 
$\mathcal{G}^{S}=\mbox{Cyc}\left(1,\frac{1}{2} T_{(6,3)} 
+\frac{1}{2} T_{(10,2)} - \langle \bullet \rangle \right)$ and the 
structure of the whole graph is given by $ \G{T_{30}}{\F_{739}} = 
\mbox{Cyc}(20,T_{(6,3)}) \oplus 2\times \mbox{Cyc}(9, T_{(10,2)}) 
\oplus  \mbox{Cyc}\left(1,\frac{1}{2} T_{(6,3)} +\frac{1}{2} T_{(10,2)} 
- \langle \bullet \rangle \right)$.

\begin{figure}[h]
\centering
\includegraphics[width=0.9\textwidth]{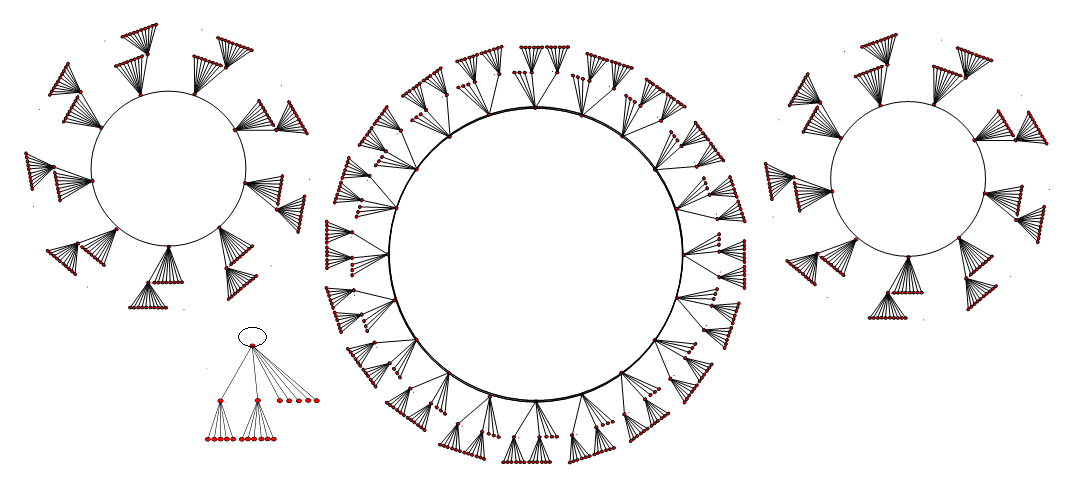}
\caption{Structure of the functional graph $\G{T_{30}}{\F_{739}}$.}
\label{FigCheby30q739}
\end{figure} 
\end{example}

\begin{example}
We consider the action of Chebyshev polynomials over the binary 
field $\F_{16}$. Using Theorem \ref{ThMain} we obtain the structure 
of the rational component $\mathcal{G}^{R}$, the quadratic component 
$\mathcal{G}^{Q}$ and the special component $\mathcal{G}^{S}$ of the 
Chebyshev graph $\GT{\F_{16}}$ for $2\leq n \leq 10$ and $n=15,17,34$ 
and $255$; see Table \ref{TableStructure}.
\begin{table}
\begin{center}
 \begin{tabular}{|c|c|c|c|} 
 \hline
 $n$ & $\mathcal{G}^{R}$ & $\mathcal{G}^{Q}$ & $\mathcal{G}^{S}$ \\ \hline 
 $2$ & $\cyc{1}\oplus\cyc{2}\oplus\cyc{4}$ & $2\times \cyc{4}$ & $\cyc{1}$  \\ \hline
 $3$ & $\cycc{2}{2}$ & $\cyc{8}$ & $\mbox{Cyc}(1,\langle\bullet \rangle)$  \\  \hline
 $4$ & $3\times\cyc{1}\oplus2\times\cyc{2}$ & $4\times\cyc{2}$ & $\cyc{1}$  \\  \hline
 $5$ & $\cycc{1}{4}$ & $\cyc{8}$ & $\cycc{1}{2}$  \\  \hline
 $6$ & $2\times\cycc{1}{2}$ & $\cyc{8}$ & $\mbox{Cyc}(1,\langle\bullet \rangle)$  \\  \hline
 $7$ & $\cyc{1}\oplus\cyc{2}\oplus\cyc{4}$ & $\cyc{8}$ & $\cyc{1}$  \\  \hline
 $8$ & $\cyc{1}\oplus\cyc{2}\oplus\cyc{4}$ & $2\times \cyc{4}$ & $\cyc{1}$  \\  \hline
 $9$ & $2\times\cycc{1}{2}$ & $2\times\cyc{4}$ & $\mbox{Cyc}(1,\langle\bullet \rangle)$   \\  \hline
 $10$ & $\cycc{1}{4}$ & $\cyc{8}$ & $\cycc{1}{2}$  \\  \hline
 $15$ & $\emptyset$ & $2\times\cyc{4}$ & $\cycc{1}{7}$  \\  \hline
 $17$ & $\cyc{1}\oplus\cyc{2}\oplus\cyc{4}$ & $\emptyset$ & $\cycc{1}{8}$  \\  \hline
 $34$ &  $3\times\cyc{1}\oplus2\times\cyc{2}$ & $\emptyset$ & $\cycc{1}{8}$   \\  \hline
 $255$ & $\emptyset$ & $\emptyset$ & $\cycc{1}{15}$ \\  \hline
\end{tabular}
\caption{Graph structure for Chebyshev polynomials $T_n$ over the 
binary field $\mathbb{F}_{16}$. 
We recall that $T=\langle m \times \bullet \rangle$ denotes a tree consisting of
a root with $m$ predecessors.}\label{TableStructure}
\end{center}
\end{table}
\end{example}

\subsection{Chebyshev involutions and permutations}

It is well known that the Chebyshev polynomial $ T_{n}$ is a permutation 
polynomial over $\F_{q}$ if and only if $\gcd(q^2-1,n)=1$. Using that 
$T_{\nu(n)}=\bullet$ if and only if $\nu=1$, this condition can be 
obtained as a direct corollary of Theorem \ref{ThMain} together with 
the decomposition into disjoint cycles.

\begin{corollary}\label{CorCyclesCheby}
The Chebyshev polynomial $T_{n}$ is a permutation polynomial over $\F_{q}$ 
if and only if $\gcd(q^2-1,n)=1$. In this case, if $q-1=\nu_{0}\omega_{0}$ 
and $q+1=\nu_{1}\omega_{1}$ are their $n$-decompositions, we have the 
following decomposition of $\GT{\Fq}$ into disjoint cycles:
$$\bigoplus_{\substack{d\mid \omega_0\\  d> 2}} 
\frac{\varphi(d)}{2\tilde{o}_{d}(n)}\times\mbox{Cyc}
\left(\tilde{o}_{d}(n),\bullet \right) \oplus 
\bigoplus_{\substack{d\mid \omega_1\\  d> 2}} 
\frac{\varphi(d)}{2\tilde{o}_{d}(n)}\times\mbox{Cyc}
\left(\tilde{o}_{d}(n),\bullet \right) \oplus k\times 
\mbox{Cyc}(1,\bullet),$$ where $k=2$ if $nq$ is odd, and $k=1$ otherwise.
\end{corollary}

A particular case of cryptographic interest is permutation polynomials 
that are involutions \cite{CMS16a,CMS16}, that is, when the 
composition with itself is the identity map. For Chebyshev polynomials 
we obtain the following characterization.

\begin{corollary}\label{CorChebyInvolutions}
Let $q-1=\nu_0\omega_0$ and $q+1=\nu_1\omega_1$ be the $n$-decomposition 
of $q-1$ and $q+1$, respectively. The Chebyshev polynomial $T_n$ is an 
involution over $\Fq$ if and only if $\nu_0=\nu_1=1$, 
$n^2 \equiv \pm 1 \pmod{\omega_1}$ and $n^2 \equiv \pm 1 \pmod{\omega_2}$.
\end{corollary}

\begin{proof}
The condition $\nu_0=\nu_1=1$ is equivalent to $\gcd(q^2-1,n)=1$ 
which is equivalent to $T_n$ being a permutation by Corollary
\ref{CorCyclesCheby}. If this condition is satisfied, $T_n$ is an 
involution if and only if $\tilde{o}_{d}(n)\in\{1,2\}$ for all 
$d$ such that $d\mid \omega_0$ or $d\mid\omega_1$, if and only if 
$n^2\equiv \pm 1$ for all $d$ with $d\mid \omega_0$ or 
$d\mid\omega_1$, if and only if $n^2 \equiv \pm 1 \pmod{\omega_1}$ 
and $n^2 \equiv \pm 1 \pmod{\omega_2}$.
\hfill $\qed$ 
\end{proof}

\begin{example} Consider the Chebyshev polynomial $T_{31}$ over $\F_{25}$. Here $n=31, q=25, \nu_0=\nu_1=1, \omega_1=24, \omega_1=26$. Since $31^2\equiv 1 \pmod{24}$ and $31^2\equiv -1 \pmod{26}$, the polynomial $T_{31}$ is an involution over $\F_{25}$.
\end{example}

\subsection{Explicit formulas for $N, T_0, C, T$ and $R$} \label{parameters}

Let $\mathcal{G} = \G{f}{X}$ be a functional graph where $X$ is 
a finite set. Given $x_0\in X$ there are integers $c\geq 1$ and 
$t\geq 0$ such that $x_0^{c+t}=x_0^{t}$. The smallest integers 
with this property are denoted by $\mbox{per}(x_0):=c$ (the 
\emph{period} of $x_0$) and $\mbox{pper}(x_0):=t$ (the 
\emph{preperiod} of $x_0$). The \emph{rho length} of $x_0$ is 
$\mbox{rho}(x_0):=\mbox{per}(x_0)+\mbox{pper}(x_0)$. We also 
consider the parameters $N, T_0, C, T$ and $R$ where
\begin{itemize}
\item $N(\mathcal{G})$ is the number of connected component of 
$\mathcal{G}$;
\item $T_0(\mathcal{G})$ is the number of periodic points; 
\item $C(\mathcal{G})= \frac{1}{|X|}\sum_{x\in X}\mbox{per}(x)$ 
is the expected value of the period;
\item $T(\mathcal{G})= \frac{1}{|X|}\sum_{x\in X}\mbox{pper}(x)$ 
is the expected value of the preperiod and
\item $R(\mathcal{G})=\frac{1}{|X|}\sum_{x\in X}\mbox{rho}(x)$ 
is the expected value of the rho length.
\end{itemize}

We apply our structural theorem to deduce explicit formulas for 
the parameters $N$, $T_0$, $C$ and $T$ for Chebyshev polynomials 
over $\F_q$ (the average rho length can be obtained from $R=C+T$). 
These parameters were studied in \cite{CS04} for the exponentiation 
map and in \cite{QPM17} for R{\'e}dei functions. 

We remark that the above parameters are invariant under isomorphism 
(i.e. isomorphic functional graphs have the same value). Related 
to $C$ and $T$ we consider the parameters $\widehat{C}$ and 
$\widehat{T}$ defined as the sum of the values of the periods 
and preperiods, respectively, from which we can easily obtain 
$C$ and $T$. The advantage of working with these parameters 
instead of $C$ and $T$ is that they are additive (i.e. 
$\widehat{C}(\mathcal{G}_1 \oplus \mathcal{G}_2)= 
\widehat{C}(\mathcal{G}_1)+\widehat{C}(\mathcal{G}_2)$ and 
$\widehat{T}(\mathcal{G}_1 \oplus \mathcal{G}_2)= 
\widehat{T}(\mathcal{G}_1)+\widehat{T}(\mathcal{G}_2)$) 
as well as the parameters $N$ and $T_0$. For additive parameters 
it suffices to know their values on each connected component. 
In the case of Chebyshev polynomials over finite fields, each 
connected component of its functional graph is uniform. It is 
immediate to check that if $\mathcal{G}=\mbox{Cyc}(m,T)$ where 
$T$ is a rooted tree with depth $D$, then $N(\mathcal{G})=1$; 
$T_0(\mathcal{G})=m$; $\widehat{C}(\mathcal{G})=m^2|T|$ and  
$\widehat{T}(\mathcal{G})=m\sum_{j=1}^{D}j h(j)$ where $h(j)$ 
denotes the number of nodes in $T$ at depth\footnote{The depth 
of a node $x$ in a rooted tree $T$ with root $r$ is the length 
of the smallest path connecting $x$ to $r$. If $T$ is a rooted 
tree attached to a cyclic node in a functional graph, the depth 
of a node is the same as its preperiod.} $j$. When the rooted 
tree $T$ is the tree attached to a $\nu$-series $T=T_{\nu(n)}$ 
we have the following formulas, whose proof is the same as the 
given one in \cite{QPM17} for R\'edei functions.

\begin{lemma}[\cite{QPM17}, Proposition 2.2.]\label{LemmaFormulasForCyc}
Let $n,\nu, m$ be positive integers with $\rad$. Consider 
$\nu(n)=(\nu_1,\nu_2,\ldots,\nu_D)$ and $\mathcal{G}=
\mbox{Cyc}(m,T_{\nu(n)})$. Then $N(\mathcal{G})= 1, 
T_0(\mathcal{G})=m, \widehat{C}(\mathcal{G})=m^2\nu$ and 
$\widehat{T}(\mathcal{G}) = m\sum_{j=1}^{D-1}\nu_1\cdots\nu_{j}$.
\end{lemma}

The next lemma shows how the parameter $\widehat{T}$ behaves 
regarding to addition and bisection of trees.

\begin{lemma}\label{LemmaTbehavior}
The following statements hold.
\begin{enumerate}
\item If $\mathcal{G}_1=\mbox{Cyc}(1,T_1)$, 
$\mathcal{G}_2=\mbox{Cyc}(1,T_2)$ and 
$\mathcal{G}=\mbox{Cyc}(1,T_1+T_2)$, then 
$\widehat{T}(\mathcal{G}) = \widehat{T}(\mathcal{G}_1) 
+ \widehat{T}(\mathcal{G}_2)$.
\item If $\mathcal{G} = \mbox{Cyc}(1,T)$ where $T$ is an even or 
quasi-even rooted tree and $\mathcal{G}'=\mbox{Cyc}(1,\frac{1}{2} T)$, 
then $\widehat{T}(\mathcal{G'}) = \left\{ \begin{array}{ll} 
\frac{\widehat{T}(\mathcal{G})}{2} & \textrm{if $T$ is even;} \\ 
\frac{\widehat{T}(\mathcal{G})+1}{2} & \textrm{if $T$ is quasi-even.}
\end{array}  \right.$
\end{enumerate}
\end{lemma}

\begin{proof}

\noindent 1. Denote by $h_1(j)$, $h_2(j)$ and $h(j)$ the number of 
nodes at depth $j$ in $T_1$, $T_2$ and $T_1+T_2$, respectively. 
Clearly we have $h(0)=1$ and $h(j)=h_1(j)+h_2(j)$ for $j\geq 1$, 
from which we obtain $\widehat{T}(\mathcal{G}) = \sum j h(j) = 
\sum j h_1(j)+ \sum j h_2(j) = \widehat{T}(\mathcal{G}_1) + 
\widehat{T}(\mathcal{G}_2)$.

\noindent 2. First we consider the case when $T$ is even. We can write $T= \langle 2 \times S \rangle$ for some forest $S$. We denote by $h_S(j)$ the number of nodes at depth $j$ in $\frac{1}{2}T = \langle S \rangle$. We have that $\widehat{T}(\mathcal{G}) = \sum j\cdot 2h_S(j) = 2 \sum j h_S(j)=2 \widehat{T}(\mathcal{G'})$. Now we consider the case when $T$ is quasi-even. We can write $T=\langle 2 \times S \oplus \langle 2 \times R\rangle \rangle$. We denote by $h_S(j)$ and $h_R(j)$ the number of nodes at depth $j$ in $\langle S \rangle$ and $\langle R\rangle$, respectively. We have that $\widehat{T}(\mathcal{G}) = (\sum j \cdot 2h_R(j)) + 1 + \sum (j+1)\cdot 2h_S(j)$ and $\widehat{T}(\mathcal{G'}) = (\sum j h_R(j)) + 1 + \sum (j+1) h_S(j)$. Thus $2 \widehat{T}(\mathcal{G'}) = \widehat{T}(\mathcal{G})+1$.
\hfill $\qed$ 
\end{proof}

Next we calculate formulas for $\widehat{C}$ and $\widehat{T}$ for the special component $\mathcal{G}^S$ of the Chebyshev functional graph $\GT{\Fq}$.

\begin{lemma}\label{LemmaGSformulas}
Let $n$ be a positive integer, $q-1=\nu_0\omega_0$ and $q+1=\nu_1\omega_1$ be the $n$-decompositions of $q-1$ and $q+1$, respectively. Let $\nu_0(n)=(a_1,\ldots,a_D)$, $\nu_1(n)=(b_0,\ldots,b_{D'})$, $A=\s{a}{D}$ and $B=\s{b}{D'}$. Denote by $\mathcal{G}^{S}$ the special component of the Chebyshev graph $\GT{\Fq}$. The following formulas for $\widehat{C}$ and $\widehat{T}$ hold.
$$\widehat{C}(\mathcal{G}^{S}) =\left\{ \begin{array}{ll} \nu_0+\nu_1, & \textrm{if $nq$ is odd;} \\
\frac{\nu_0+\nu_1}{2}, & \textrm{otherwise.} \end{array} \right.
 \quad \textrm{and} \quad      \widehat{T}(\mathcal{G}^{S}) =\left\{ \begin{array}{ll} A+B, & \textrm{if $nq$ is odd;} \\
\frac{A+B}{2}, & \textrm{otherwise.} \end{array} \right.$$
\end{lemma}

\begin{proof}
First we consider the case when $qn$ is odd. In this case both $\nu_0$ 
and $\nu_1$ are odd and, by Proposition \ref{PropTreeVseriesEven}, 
both rooted trees $T_{\nu_0(n)}$ and $T_{\nu_1(n)}$ are even. From 
Proposition \ref{PropBisectionCardinality}, Theorem \ref{ThMain} and 
the fact that $|T_{\nu(n)}|=\nu$ (see Equation (\ref{TreeAssociatedEq}) 
and the following paragraph), we have
$\widehat{C}(\mathcal{G}^{S}) = 2 |\frac{1}{2} T_{\nu_0(n)} 
+ \frac{1}{2} T_{\nu_1(n)}| = 2 \left( \frac{\nu_0+1}{2} 
+ \frac{\nu_1+1}{2} -1\right) = \nu_0 + \nu_1$. 
Applying Lemmas \ref{LemmaFormulasForCyc} and \ref{LemmaTbehavior} 
we obtain $\widehat{T}(\mathcal{G}^{S}) = 2\cdot(\frac{A}{2} + 
\frac{B}{2})= A+B$.

Now we consider the case when $q$ is even. In this case again both 
$\nu_0$ and $\nu_1$ are odd and consequently both rooted trees 
$T_{\nu_0(n)}$ and $T_{\nu_1(n)}$ are even. By Proposition \ref{PropBisectionCardinality} and Theorem \ref{ThMain}, we have
$\widehat{C}(\mathcal{G}^{S})= |\frac{1}{2} T_{\nu_0(n)} 
+ \frac{1}{2} T_{\nu_0(n)}| = \frac{\nu_0+1}{2} 
+ \frac{\nu_1+1}{2}-1 = \frac{\nu_0+\nu_1}{2}$. 
Applying Lemmas \ref{LemmaFormulasForCyc} and \ref{LemmaTbehavior} 
we obtain $\widehat{T}(\mathcal{G}^{S}) = \frac{A}{2} + 
\frac{B}{2}= \frac{A+B}{2}$.

The remainder case is when $n$ is even and $q$ is odd. In this 
case both $\nu_0$ and $\nu_1$ are even. By Proposition \ref
{PropTreeVseriesEven} both $T_{\nu_0(n)}$ and $T_{\nu_1(n)}$ are 
quasi-even. By Proposition \ref{PropBisectionCardinality} and 
Theorem \ref{ThMain} we have
$\widehat{C}(\mathcal{G}^{S})= |\frac{1}{2} T_{\nu_0(n)} 
+ \frac{1}{2} T_{\nu_0(n)} - \langle \bullet \rangle| 
= \frac{\nu_0+2}{2} + \frac{\nu_1+2}{2}-1-1 
= \frac{\nu_0+\nu_1}{2}$. Applying Lemmas \ref{LemmaFormulasForCyc} 
and \ref{LemmaTbehavior} we obtain $\widehat{T}(\mathcal{G}^{S}) = 
\frac{A+1}{2} + \frac{B+1}{2} -1= \frac{A+B}{2}$.
\hfill $\qed$ 
\end{proof}

\begin{theorem}
Let $n$ be a positive integer. Let $q-1=\nu_0\omega_0$ and $q+1=\nu_1
\omega_1$ be the $n$-decompositions of $q-1$ and $q+1$, respectively. 
Let $\nu_0(n)=(a_1,\ldots,a_D)$ and $\nu_1(n)=(b_0,\ldots,b_{D'})$. 
Then, the following holds for $\mathcal{G}=\GT{\Fq}$:
\begin{itemize}
\item[$\bullet$] the number of cycles in $\mathcal{G}(T_n/\F_q)$ is $N(\mathcal{G})= \frac{1}{2}(\sum_{d\mid \omega_0} \frac{\varphi(d)}{\tilde{o}_{d}(n)} + \sum_{d\mid \omega_1} \frac{\varphi(d)}{\tilde{o}_{d}(n)})$;
\item[$\bullet$] the number of periodic points is given by $T_0(\mathcal{G})=\frac{\omega_0+\omega_1}{2}$;
\item[$\bullet$] the expected value of $\mbox{per}(a)$ where $a$ runs over the elements of $\F_q$ is \\
$C(\mathcal{G})=\frac{q-1}{2q}(\frac{1}{\omega_0}\sum_{d\mid \omega_0}\varphi(d)\tilde{o}_{d}(n)) + \frac{q+1}{2q}(\frac{1}{\omega_1}\sum_{d\mid \omega_1}\varphi(d)\tilde{o}_{d}(n))$;
\item[$\bullet$] the expected value of $\mbox{pper}(a)$ where $a$ runs over the elements of $\F_q$ is \\
$T(\mathcal{G})= \frac{q-1}{2q}(\frac{1}{\nu_0}\sum_{i=1}^{D-1}a_1\ldots a_i) + \frac{q+1}{2q}(\frac{1}{\nu_1}\sum_{i=1}^{D'-1}b_1\ldots b_i)$.
\end{itemize}
\end{theorem}

\begin{proof}
Applying Theorem \ref{ThMain} we have 
\begin{align}\label{EqOne1}
N(\mathcal{G})&=N(\mathcal{G}^{R})+ N(\mathcal{G}^{Q}) + N(\mathcal{G}^{S}) \nonumber \\  &= \sum_{\substack{d\mid \omega_0\\  d> 2}} \frac{\varphi(d)}{2\tilde{o}_{d}(n)} +  \sum_{\substack{d\mid \omega_1\\  d> 2}} \frac{\varphi(d)}{2\tilde{o}_{d}(n)} +  \left\{ \begin{array}{ll}
2, & \textrm{ if $nq$ is odd;}\\
1, & \textrm{ if $nq$ is even.}
\end{array}  \right.
\end{align}
Since both $\omega_0$ and $\omega_1$ are even when $nq$ is odd and both $\omega_0$ and $\omega_1$ are odd when $nq$ is even, we have
\begin{equation}\label{EqTwo1}
\sum_{\substack{d\mid \omega_0\\  d\leq 2}} \frac{\varphi(d)}{2\tilde{o}_{d}(n)} +  \sum_{\substack{d\mid \omega_1\\  d\leq  2}} \frac{\varphi(d)}{2\tilde{o}_{d}(n)} =  \left\{ \begin{array}{ll}
1+1=2, & \textrm{ if $nq$ is odd;}\\
\frac{1}{2}+\frac{1}{2} = 1, & \textrm{ if $nq$ is even.}
\end{array}  \right.
\end{equation}
By Equations (\ref{EqOne1}) and (\ref{EqTwo1}) we have $N(\mathcal{G})= \sum_{d\mid \omega_0} \frac{\varphi(d)}{2\tilde{o}_{d}(n)} + \sum_{d\mid \omega_1} \frac{\varphi(d)}{2\tilde{o}_{d}(n)}$.


Applying Theorem \ref{ThMain} we have 
\begin{align}\label{EqOne2}
T_0(\mathcal{G})&=T_0(\mathcal{G}^{R})+ T_0(\mathcal{G}^{Q}) + T_0(\mathcal{G}^{S}) \nonumber \\  &= \sum_{\substack{d\mid \omega_0\\  d> 2}} \frac{\varphi(d)}{2\tilde{o}_{d}(n)} \cdot \tilde{o}_{d}(n) +  \sum_{\substack{d\mid \omega_1\\  d> 2}} \frac{\varphi(d)}{2\tilde{o}_{d}(n)}\cdot \tilde{o}_{d}(n) +  T_0(\mathcal{G}^{S}) \nonumber  \\
&= \sum_{\substack{d\mid \omega_0\\  d> 2}} \frac{\varphi(d)}{2} +  \sum_{\substack{d\mid \omega_1\\  d> 2}} \frac{\varphi(d)}{2} +  \left\{ \begin{array}{ll}
2, & \textrm{ if $nq$ is odd;}\\
1, & \textrm{ if $nq$ is even.}
\end{array}  \right.
\end{align}
Since $\omega_0$ and $\omega_1$ are even when $nq$ is odd and $\omega_0$ and $\omega_1$ are odd otherwise, we have

\begin{equation}\label{EqTwo2}
\sum_{\substack{d\mid \omega_0\\  d\leq 2}} \frac{\varphi(d)}{2} +  \sum_{\substack{d\mid \omega_1\\  d\leq  2}} \frac{\varphi(d)}{2} =  \left\{ \begin{array}{ll}
1+1=2, & \textrm{ if $nq$ is odd;}\\
\frac{1}{2}+\frac{1}{2} = 1, & \textrm{ if $nq$ is even.}
\end{array}  \right.
\end{equation}
By Equations (\ref{EqOne2}) and (\ref{EqTwo2}) we have $$T_0(\mathcal{G})= \sum_{d\mid \omega_0} \frac{\varphi(d)}{2} + \sum_{d\mid \omega_1} \frac{\varphi(d)}{2} = \frac{\omega_0}{2} + \frac{\omega_1}{2} = \frac{\omega_0 + \omega_1}{2}.$$


Using that $\mbox{Cyc}(m,T_{\nu(n)})$ has exactly $m\nu$ nodes (see Equation (\ref{TreeAssociatedEq}) and the following paragraph) and applying Theorem \ref{ThMain} and Lemma \ref{LemmaGSformulas} we have 
\begin{align}\label{EqOne3}
\widehat{C}(\mathcal{G})&=\widehat{C}(\mathcal{G}^{R})+ \widehat{C}(\mathcal{G}^{Q}) + \widehat{C}(\mathcal{G}^{S}) \nonumber \\  &= \sum_{\substack{d\mid \omega_0\\  d> 2}} \frac{\varphi(d)}{2\tilde{o}_{d}(n)} \cdot \tilde{o}_{d}(n) \cdot \tilde{o}_{d}(n)\nu_0 +  \sum_{\substack{d\mid \omega_1\\  d> 2}} \frac{\varphi(d)}{2\tilde{o}_{d}(n)}\cdot \tilde{o}_{d}(n)\cdot \tilde{o}_{d}(n)\nu_1 +  \widehat{C}(\mathcal{G}^{S}) \nonumber  \\
&= \frac{\nu_0}{2}\sum_{\substack{d\mid \omega_0\\  d> 2}} \varphi(d)\tilde{o}_{d}(n) + \frac{\nu_1}{2}  \sum_{\substack{d\mid \omega_1\\  d> 2}} \varphi(d)\tilde{o}_{d}(n) +  \left\{ \begin{array}{ll}
\nu_0+\nu_1, & \textrm{ if $nq$ is odd;}\\
\frac{\nu_0+\nu_1}{2}, & \textrm{ otherwise.}
\end{array}  \right. \hspace{-5mm}
\end{align}
Since $\omega_0$ and $\omega_1$ are even when $nq$ is odd and $\omega_0$ and $\omega_1$ are odd otherwise, we have
\begin{equation}\label{EqTwo3}
\frac{\nu_0}{2}\sum_{\substack{d\mid \omega_0\\  d\leq 2}} \varphi(d)\tilde{o}_{d}(n) + \frac{\nu_1}{2}  \sum_{\substack{d\mid \omega_1\\  d\leq 2}} \varphi(d)\tilde{o}_{d}(n) =  \left\{ \begin{array}{ll}
  \nu_0+\nu_1, & \textrm{ if $nq$ is odd;}\\
\frac{\nu_0+\nu_1}{2}, & \textrm{ otherwise.}
\end{array}  \right. 
\end{equation}
By Equations (\ref{EqOne3}) and (\ref{EqTwo3}) we have 
\begin{align*}
\widehat{C}(\mathcal{G}) &= \frac{\nu_0}{2}\sum_{d\mid \omega_0} \varphi(d)\tilde{o}_{d}(n) +\frac{\nu_1}{2} \sum_{d\mid \omega_1} \varphi(d)\tilde{o}_{d}(n) \\ &=  \frac{q-1}{2}\left( \frac{1}{\omega_0}\sum_{d\mid \omega_0} \varphi(d)\tilde{o}_{d}(n)\right) +\frac{q+1}{2}\left( \frac{1}{\omega_1} \sum_{d\mid \omega_1} \varphi(d)\tilde{o}_{d}(n)\right).
\end{align*}
Dividing both sides by $q$ we obtain the formula for $C(\mathcal{G})$. 


Now we deduce the formula for $T$. Denote $A= \s{a}{D}$ and $B=\s{B}{D'}$. Using Lemmas \ref{LemmaFormulasForCyc} and \ref{LemmaGSformulas} and Theorem \ref{ThMain} we obtain
\begin{align}\label{EqOne4}
\widehat{T}(\mathcal{G})&=\widehat{T}(\mathcal{G}^{R})+ \widehat{T}(\mathcal{G}^{Q}) + \widehat{T}(\mathcal{G}^{S}) \nonumber \\  &= \sum_{\substack{d\mid \omega_0\\  d> 2}} \frac{\varphi(d)}{2\tilde{o}_{d}(n)} \cdot \tilde{o}_{d}(n) A+  \sum_{\substack{d\mid \omega_1\\  d> 2}} \frac{\varphi(d)}{2\tilde{o}_{d}(n)}\cdot \tilde{o}_{d}(n) B +  \widehat{T}(\mathcal{G}^{S}) \nonumber  \\
&= \frac{A}{2}\sum_{\substack{d\mid \omega_0\\  d> 2}} \varphi(d) + \frac{B}{2}  \sum_{\substack{d\mid \omega_1\\  d> 2}} \varphi(d) +  \left\{ \begin{array}{ll}
A+B, & \textrm{ if $nq$ is odd;}\\
\frac{A+B}{2}, & \textrm{ otherwise.}
\end{array}  \right. 
\end{align}
Since $\omega_0$ and $\omega_1$ are even when $nq$ is odd and $\omega_0$ and $\omega_1$ are odd otherwise, we have
\begin{equation}\label{EqTwo4}
\frac{A}{2}\sum_{\substack{d\mid \omega_0\\  d\leq 2}} \varphi(d) + \frac{B}{2}  \sum_{\substack{d\mid \omega_1\\  d\leq 2}} \varphi(d) =  \left\{ \begin{array}{ll}
  A+B, & \textrm{ if $nq$ is odd;}\\
\frac{A+B}{2}, & \textrm{ otherwise.}
\end{array}  \right. 
\end{equation}
By Equations (\ref{EqOne4}) and (\ref{EqTwo4}) we have 
\begin{align*}
\widehat{T}(\mathcal{G}) &= \frac{A}{2}\sum_{d\mid \omega_0} \varphi(d) +\frac{B}{2} \sum_{d\mid \omega_1} \varphi(d) =  \frac{A\omega_0}{2} + \frac{B\omega_1}{2}. = \frac{q-1}{2}\cdot \frac{A}{\nu_0} + \frac{q+1}{2}\cdot \frac{B}{\nu_1}.
\end{align*}
Dividing both sides by $q$ we obtain the formula for $T(\mathcal{G})$. 
\hfill $\qed$ 
\end{proof}

\begin{acknowledgements}
The first author was supported by Fapesp grant 201526420-1 and the 
second author by NSERC of Canada. Most of this work was done while 
the first author was visiting Carleton University. This author 
wishes to thank the Fields Institute for funding this visit. 
\end{acknowledgements}


\begin{thebibliography}{99}

\bibitem{BBS86} L. Blum, M. Blum and M. Shub, {\it A simple 
unpredictable pseudo-random number generator}, SIAM J. 
Comp. 15(2): 364--383 (1986).

\bibitem{CMS16a} P. Charpin, S. Mesnager and S. Sarkar, {\it 
Dickson polynomials that are involutions}, in Contemporary 
Developments in Finite Fields and Applications, World 
Scientific, 22--47 (2016).

\bibitem{CMS16}  P. Charpin, S. Mesnager and S. Sarkar, 
{\it Involutions over the Galois field ${\mathbb F}_{2^n}$}, 
IEEE Transactions on Information Theory 62: 2266--2276 (2016).

\bibitem{CS04}
W. Chou and I.E. Shparlinski, On the cycle structure of 
repeated exponentiation modulo a prime, J. Number 
Theory 107 (2004), 345--356.

\bibitem{FL16}
S. Fan and L. Liao, Dynamical structures of Chebyshev 
polynomials on $\Z_2$, J. Number Theory 169 (2016), 
174--182.

\bibitem{FlajOdl1990}
Ph. Flajolet and A. M. Odlyzko, Random mapping statistics,
In Advances in Cryptology EUROCRYPT 89, Lect. Notes 
in Comp. Science 434 (1990), 329--354.

\bibitem{Gassert14}
T. A. Gassert, Chebyshev action on finite fields, Discr. 
Math. 315-316 (2014) 83--94. 


\bibitem{L30} D. H. Lehmer, {\it An extended theory of 
Lucas functions}, Ann. Math. 31(3): 419--448 (1930).

\bibitem{L78} E. Lucas, {\it Th\'eorie des fonctions 
num\'eriques simplement p\'eriodiques}, Amer. J. Math. 1(4):
289--321 (1878).

\bibitem{LMT1993}
R. Lidl, G. L. Mullen and G. Turnwald, Dickson polynomials. 
Vol. 65. Chapman \& Hall/CRC, 1993.

\bibitem{HFF} G.~Mullen and D.~Panario, Handbook of Finite 
Fields, Discr. Math. and Its Applications, Chapman 
and Hall/CRC, 2013.

\bibitem{P75} J. M. Pollard, {\it A Monte Carlo method for 
factorization}, BIT 15(3): 331--334 (1975).

\bibitem{P78} J. M. Pollard, {\it Monte Carlo methods for 
index computation (mod $p$)}, Math. Comp. 32(143): 918--924 (1978).

\bibitem{QP15}
C. Qureshi and D. Panario, R\'edei actions on finite fields
and multiplication map in cyclic groups, SIAM J. on 
Discr. Math. 29 (2015), 1486--1503.

\bibitem{QPM17} C. Qureshi, D. Panario and R. Martins, 
{\it Cycle structure of iterating R\'edei functions}, 
Advances in Mathematics of Communications 11(2): 397-407 (2017)


\bibitem{Ugolini}
S. Ugolini. On the iterations of certain maps 
$X\mapsto K\cdot(X+X^{-1})$ over finite fields of odd 
characteristic, J. of Number Theory 142 (2014), 274--297.

\bibitem{VS04}
T. Vasiga and J. Shallit, On the iteration of certain quadratic 
maps over $GF(p)$, Discr. Math. 277 (2004), 219--240.

\end{thebibliography}
\end{document}